\documentclass[10pt, twocolumn]{article}

\usepackage[a4paper,margin=0.75in]{geometry}
\usepackage{graphicx}
\usepackage{amsmath}
\usepackage{amssymb}
\usepackage{mathtools}
\usepackage{subcaption}
\usepackage{siunitx}
\usepackage{xcolor}
\usepackage{url}
\usepackage{authblk}

\newcommand{\bm}{\boldsymbol}
\newtheorem{theorem}{Theorem}
\newtheorem{lemma}{Lemma}
\newtheorem{proposition}{Proposition}
\newtheorem{definition}{Definition}

\newtheorem{remark}{Remark}
\newenvironment{proof}{\par\noindent\textit{Proof.}\ }{\hfill$\square$\par}

\title{Angle-based Localization and Rigidity Maintenance Control for Multi-Robot Networks} 

\author[1]{J. Francisco Presenza
\thanks{
This work was partially supported by Universidad de Buenos Aires [grant number UBACyT 20020220100053BA]. L. Colombo acknowledge financial support from Grant PID2022-137909NB-C21 funded by MCIN/AEI/ 10.13039/501100011033 and iRoboCity2030-CM, Robótica Inteligente para Ciudades Sostenibles (TEC-2024/TEC-62). 
Corresponding author J.~Francisco Presenza.}
} 
\author[2]{Leonardo J. Colombo}
\author[3]{Ignacio Mas}
\author[3]{Juan I. Giribet}

\affil[1]{\normalsize Institute of Engineering Technology and Sciences ``Hilario Fern\'andez Long'' (CONICET-UBA) - Av. Paseo Col\'on 850, C1063ACV, Buenos Aires, Argentina.}  
\affil[2]{\normalsize Centre for Automation and Robotics (CSIC-UPM) - Ctra. M300 Campo Real, Km 0.200, Arganda del Rey, 28500, Madrid, Spain.}             
\affil[3]{\normalsize Artificial Intelligence and Robotics Laboratory (UDESA) and CONICET - Vito Dumas 284, B1644BID, Victoria, Argentina.}        

\date{}

\begin{document}
\maketitle

\begin{abstract}
In this work, we study angle-based localization and rigidity maintenance control for multi-robot networks.
First, we establish the relationship between angle rigidity and bearing rigidity considering \textit{directed} sensing graphs and \textit{body-frame} bearing measurements in both $2$ and $3$-\textit{dimensional space}.
In particular, we demonstrate that a framework in $\mathrm{SE}(d)$ is infinitesimally bearing rigid if and only if it is infinitesimally angle rigid and each robot obtains at least $d-1$ bearing measurements ($d \in \{2, 3\}$).
Building on these findings, this paper proposes a distributed angle-based localization scheme and establishes local exponential stability under switching sensing graphs, requiring only infinitesimal angle rigidity across the visited topologies.
Then, since the set of available angles strongly depends on the robots' spatial configuration due to sensing constraints, we investigate rigidity maintenance control.
The \textit{angle rigidity eigenvalue} is presented as a metric for the degree of rigidity.
A decentralized gradient-based controller capable of executing mission-specific commands while maintaining a sufficient level of angle rigidity is proposed.
Simulations were conducted to evaluate the scheme's effectiveness and practicality.
\end{abstract}

\section{Introduction}
\label{sec:introduction}

Cooperative localization is a fundamental problem in multi-robot systems, enabling the decentralized estimation of robot positions and orientations using only inter-agent measurements.
Over the past decade, extensive research has focused on \textit{bearing-based} methods, which exploit line-of-sight direction measurements between robots.
Recently, \textit{angle-based} approaches, which employ the angles between pairs of measured bearings, have emerged as a particularly attractive alternative.
\textit{Rigidity theory} studies whether the available measurements (bearings or angles) are sufficient to infer the robots’ positions and, depending on the context, the orientations as well.

Early approaches to bearing-based localization commonly assumed that all measurements were expressed in a common reference frame, implying that each agent knows its orientation.
This allows for the modeling of the robot states as elements of $\mathbb{R}^d$ and the use of undirected sensing networks, since the two bearings corresponding to a pair of robots carry the same information.
Such a setup resulted in protocols with global convergence guarantees under mild conditions, e.g., containing a spanning Laman subgraph; see \cite{Zhao2019CSM} and references therein.
However, this assumption can be hard to meet in practice since the sensors used for absolute orientation estimation, such as compasses, are either unavailable or strongly affected by magnetic disturbances.
Additionally, due to sensor limitations, such as narrow camera fields of view, bearing measurements are generally not bidirectional; that is, robot $i$ may observe robot $j$, but not vice versa.
For these reasons, real-world multi-robot systems must handle unknown robot poses in $\mathrm{SE}(d)$, along with directed sensing graphs.

Following this, many works have proposed exploiting the body-frame bearing measurements for orientation estimation.
For example, \cite{Piovan2013AUT} restricts the robot poses to $\mathrm{SE}(2)$ and considers that all bearing measurements are bidirectional---rendering an undirected sensing graph.
Under these conditions, orientations can be estimated from bearings, provided that the graph is connected.
In $\mathrm{SE}(3)$, bidirectional bearings and graph connectivity are not sufficient; hence, \cite{VanTran2018CDC,Leonardos2019CDC,VanTran2020AUT,Cao2021AUT,Boughellaba2023CSL} require more involved sensing topologies, e.g., some pairs of robots with bidirectional bearings and common neighbors.
Other schemes consider directed graphs under the more general condition of infinitesimal bearing rigidity, but these approaches are restricted to $\mathrm{SE}(2)$ \cite{Zelazo2014ECC} and $\mathbb{R}^3 \times \mathbb{S}^1$ \cite{Schiano2016IROS}.
The primary limitation of these approaches arises from the constraints associated with estimating orientations from bearings.
Indeed, they require at least $d-1$ observed neighbors per robot, constraining the robots' freedom of motion, as all cameras must be engaged in mutual sensing.

Subsequent strategies assumed the availability of relative orientation measurements to estimate all orientations in a common frame \cite{Trinh2018CCTA,Li2020TCNS,Boughellaba2025TAC}.
Once the orientations are recovered, common-frame bearing-based localization can be implemented to estimate the robots positions with global convergence guarantees. 
While an appealing alternative, this approach introduces an extra layer of complexity since the estimation must be performed on the manifold $\mathrm{SE}(d)$ and because orientation measurement noise propagates into the position estimates.

More recently, motivated by these challenges, several works have focused on angle-based localization; see \cite{Jing2021TAC,Chen2022AUTa,Chen2022AUTb,Chen2022TSP}.
These angles are independent of the reference frame; hence, orientation estimation procedures can be avoided, simplifying the description and reducing the system’s dimensionality.
As a consequence, angle-based localization does not require a minimum number of bearing measurements per robot (some robots might not acquire bearing measurements at all), thus increasing their freedom of motion. 
Although existing work provides solid theoretical foundations, several shortcomings prevent their use in real-world applications.
For example, the localization schemes found in \cite{Jing2021TAC,Chen2022AUTa,Chen2022AUTb,Chen2022TSP} are only applicable in $\mathbb{R}^2$ or require special topologies such as triangular and tetrahedral in $\mathbb{R}^2$ and $\mathbb{R}^3$, respectively.
Developing localization schemes that allow for arbitrary directed and time-varying topologies can significantly improve the applicability of these systems.

Despite the close connection between bearing-based and angle-based coordination, their relationship remains poorly understood, except under restrictive assumptions. 
For example, recent works have established the equivalence between bearing rigidity and angle rigidity \cite{Jing2019AUT} or signed angle rigidity \cite{Jinpeng2025CDC} but only in $\mathbb{R}^2$.
Additionally, these findings assume that all bearings are measured in a common frame, constraining the analysis to undirected graphs.
In contrast, the results developed in this paper consider \textit{directed} sensing graphs and \textit{body-frame} bearing measurements in both $2$ and $3$-dimensional space.
In particular, as our first main contribution, we show that in $\mathrm{SE}(d)$ ($d \in \{2, 3\}$), infinitesimal bearing rigidity holds if and only if the framework is infinitesimally angle rigid and each robot obtains at least $d-1$ bearing measurements. 
To our knowledge, this is the first formal equivalence between bearing rigidity and angle rigidity that naturally applies to vision-based multi-robot systems since it (a) does not require bidirectional bearing measurements; (b) deals with body-frame measurements; and (c) works beyond the plane.
In addition, as our second main contribution, we propose a distributed angle-based localization scheme and establish local exponential stability under directed switching topologies, provided that all visited graphs are infinitesimally angle rigid.

Finally, we design a control scheme that enables robots to follow mission-related trajectories while guaranteeing infinitesimal angle rigidity.
Indeed, as our third main contribution, we extend the rigidity maintenance techniques proposed in \cite{Zelazo2015IJRR,Schiano2017ICRA} to the case of angle measurements.
We present the angle rigidity eigenvalue and propose its use as a metric for the degree of rigidity.
A decentralized gradient-based controller capable of executing mission-specific commands while maintaining a sufficient level of rigidity is proposed.
The controller preserves angle rigidity by keeping the rigidity eigenvalues above a threshold.
We note that localizability conditions less restrictive than rigidity have been studied in \cite{DeCarli2023MRS} and \cite{Colombo2019CDC} for the cases of bearing and distance measurements, respectively. 
However, these works do not aim to guarantee a minimum level of localizability, which limits their applicability under sensing constraints.
To validate the proposed approach, a mission involving multiple target tracking is proposed, which produces changes in the sensing topology that challenge the rigidity condition and test the capabilities of the controller.
Simulations were performed to evaluate the effectiveness of the scheme.

The remainder of this paper is organized as follows. Section \ref{sec:preliminaries} presents basic definitions and an overview of the bearing and angle rigidity theory in $\mathrm{SE}(d)$. 
Section \ref{sec:bearing_vs_angle} establishes the equivalence between bearing rigidity and angle rigidity. 
The angle-based position estimator and the rigidity maintenance controller are introduced in Section \ref{sec:control}, together with a discussion of their convergence and decentralized implementation. 
The simulation results are presented in Section \ref{sec:validation}, and Section \ref{sec:conclusions} provides final conclusions and directions for future work.

\section{Preliminaries and Notation}
\label{sec:preliminaries}

Lowercase letters without a superscript, e.g., $x \in \mathbb{R}^d$ are used to express a vector in a common reference frame associated with the standard basis $\{e_1, \dots, e_d\} \subset \mathbb{R}^d$.
A column stack of $n$ vectors $v_1, \ldots, v_n$ is represented as $[v_1; \dots; v_n]$ and also as $[v_i]_{i \in \mathcal{I}}$ given $\mathcal{I} = \{1, \ldots, n\}$.
The vector of all zeros and all ones are $\bm{0}_n, \bm{1}_n \in \mathbb{R}^n$, respectively, and the identity matrix is $\bm{I}_n \in \mathbb{R}^{n \times n}$.
The null space of a matrix is $\mathrm{Null}(A)$, and the Kronecker product is symbolized as $A \otimes B$.
The unit $d$-sphere is $\mathbb{S}^{d-1} = \{v \in \mathbb{R}^d : \| v \| = 1\}$.

Let $\mathcal{M}$ be a smooth manifold, then the tangent and normal spaces at $x \in \mathcal{M}$ are denoted as $T_x \mathcal{M}$ and $N_x \mathcal{M}$, respectively.
The gradient at $x$ of a differentiable scalar-valued function $f : \mathcal{M} \to \mathbb{R}$ is defined as the unique element of $T_x \mathcal{M}$ that satisfies
\begin{equation*}
    \langle \nabla_x f(x), v \rangle = (\mathrm{d}_x f) (v), \ \forall v \in T_x \mathcal{M},
\end{equation*}
where $\mathrm{d}_xf$ is the differential of $f$ at $x$ and $\langle \cdot, \cdot \rangle$ denotes a Riemannian metric on $\mathcal{M}$.
For more details, see \cite{Lee2013ISM}.

Let $\mathrm{SE}(d) = \mathbb{R}^d \times \mathrm{SO}(d)$ be the special Euclidean group.
Here $\mathrm{SO}(d)$ corresponds to the Lie group of rotation matrices,
and its Lie algebra $\mathfrak{so}(d)$ denotes the space of skew-symmetric matrices, both of dimension $d' \coloneqq \binom{d}{2}$. 
Now, let $S: \mathbb{R}^{d'} \to \mathfrak{so}(d)$ be the invertible linear map such that for all $x \in \mathbb{R}^{d'}$ and $y \in \mathbb{R}^d$, $S(x) y = x [-y_2; y_1]$ when $d=2$, and $S(x) y = x \times y$ when $d=3$.
It follows that $T_R\mathrm{SO}(d)=\{S(\omega) R : \omega \in \mathbb{R}^{d'}\}$.

In this work, the sensing network is modeled by a \textit{directed graph} denoted as $\mathcal{G} = (\mathcal{V}, \mathcal{E})$ with a vertex set $\mathcal{V} = \{1, \ldots, N\}$ and an edge set $\mathcal{E} \subseteq \mathcal{V} \times \mathcal{V}$.
Here, $(i, j) \in \mathcal{E}$ means that agent $i$ detects agent $j$ with its bearing sensor.
The use of directed graphs allows us to account for the non-reciprocity of the bearing measurements.
The set of (outgoing) neighbors of $i$ is $\mathcal{O}_i = \{j : (i, j) \in \mathcal{E}\}$.
Let $(p_i, R_i)$ represent the pose of robot $i$, with position $p_i \in \mathbb{R}^d$ and orientation $R_i \in \mathrm{SO}(d)$; and $(\bm p, \bm R)$ with  $\bm p = [p_i]_{i \in \mathcal{V}} \in \mathbb{R}^{d|\mathcal{V}|}$ and $\bm R = [R_i]_{i \in \mathcal{V}} \in \mathrm{SO}(d)^{|\mathcal{V}|}$ represents the joint configuration.

A \textit{framework} is a pair consisting of a graph and a joint pose. Depending on the context, it may be written as $(\mathcal G, (\bm p, \bm R))$ when orientations are included, or as $(\mathcal G, \bm p)$ when only positions are considered.
The body-frame bearing acquired when $j \in \mathcal{O}_i$ is described by the unit vector $\beta^i_{ij} \coloneqq R_i^\top \beta_{ij}$ where $\beta_{ij} \coloneqq p_{ij} / d_{ij}$, $p_{ij} \coloneqq p_j - p_i$, and $d_{ij} \coloneqq \|p_{ij}\|$.
Throughout this work, we assume that $p_i \neq p_j$ for all $i, j \in \mathcal{V}$.
The (cosine of the) angle measurement obtained when $\{j, k\} \subseteq \mathcal{O}_i$ is $\alpha_{ijk} \coloneqq \beta_{ij}^\top \beta_{ik} = (\beta^i_{ij})^\top \beta^i_{ik}$.
The set of all angle indices is $\mathcal{A} \coloneqq \{(i, j, k) : \{j, k\} \subseteq \mathcal{O}_i, j < k\}$.
Note that we define $\mathcal A$ to include all angles induced by a given edge set $\mathcal E$, as this is the most natural way to represent observed angles.

\subsection{Bearing Rigidity}

Bearing rigidity in $\mathrm{SE}(d)$ studies the conditions under which robot poses can be reconstructed from body-frame bearing measurements up to a similarity transformation; that is, any combination of uniform translations, rotations, reflections, and scalings. 

Consider a framework $(\mathcal G, (\bm p, \bm R))$ over $\mathrm{SE}(d)$, $d \in \{2, 3\}$.
Let $\bm{\beta}_{\mathcal G} : (\mathrm{SE}(d))^{|\mathcal{V}|} \to (\mathbb{S}^{d-1})^{|\mathcal{E}|}$ be such that $\bm{\beta}_{\mathcal G}(\bm p, \bm R) = (\beta^i_{ij})_{(i, j) \in \mathcal{E}}$ is denoted as the \textit{bearing function}.
The \textit{bearing rigidity matrix} $B_{\mathcal G}(\bm p, \bm R) \in \mathbb{R}^{d|\mathcal{E}| \times (d+d')|\mathcal{V}|}$ is the differential of $\bm{\beta}_{\mathcal{G}}$, which maps robot linear and angular velocities to bearing velocities,
\begin{equation*}
    \dot{\bm \beta}_{\mathcal G} = B_{\mathcal G}(\bm p, \bm R) \begin{bmatrix}\bm{v} \\ \bm{\omega}\end{bmatrix},
\end{equation*}
where $\bm{v} = [v_i]_{i \in \mathcal{V}} \in \mathbb{R}^{d|\mathcal{V}|}$ and $\bm{\omega} = [\omega_i]_{i \in \mathcal{V}} \in \mathbb{R}^{d'|\mathcal{V}|}$.
It is useful to make the structure of the rigidity matrix explicit.
To this end, let $B_{ij}(\bm p, \bm R)$ be the $d \times d|\mathcal{V}|$ block row corresponding to the edge $(i, j) \in \mathcal{E}$.
It follows that
\begin{equation}
    \!\dot{\beta}^i_{ij}\!=\!B_{ij} (\bm p, \bm R)\!\begin{bmatrix}\bm{v} \\ \bm{\omega}\end{bmatrix}\!=\!
    R_i^\top\!\left(\!\frac{P_{ij}}{d_{ij}} (v_j\!-\!v_i)\!-\!S(\omega_i)\beta_{ij}\!\right)\!, 
    \label{eq:bearing_rigidity_perturbation}
\end{equation}
where $P_{ij} \coloneqq \bm{I}_d - \beta_{ij} \beta_{ij}^\top$ is the orthogonal projection matrix onto $\operatorname{span}\{\beta_{ij}\}^{\perp}$.
We will also use this matrix in the robot's $i$ reference frame, i.e., $P^i_{ij} \coloneqq \bm{I}_d - \beta^i_{ij} (\beta^i_{ij})^\top$.

Consider an infinitesimal motion of $(\bm p, \bm R)$ generated by a similarity transformation. The space of all such motions is
\begin{equation}
\begin{split}
    \mathcal{T}_b(\bm p, \bm R)\!\coloneqq\!\big\{&[\bm{v}; \bm{\omega}]\!: v_i \bm\!= \kappa p_i + S(\omega) p_i + v, \ \omega_i\!= \omega, \\
    & \kappa \in \mathbb{R}, \ \omega \in \mathbb{R}^{d'}, \ v \in \mathbb{R}^d, \ \forall i \in \mathcal{V}\big \}.
\end{split}
\end{equation}
It follows that $\operatorname{dim} \mathcal{T}_b(\bm p, \bm R) = d + d' + 1$ and $\mathcal{T}_b(\bm p, \bm R) \subseteq \operatorname{Null} B_{\mathcal G}(\bm p, \bm R)$; see \cite{Michieletto2016CDC}.
We are now ready to define infinitesimal bearing rigidity in $\mathrm{SE}(d)$.

\begin{definition}
    A framework $(\mathcal G, (\bm p, \bm R))$ is said to be infinitesimally bearing rigid (IBR) if $\operatorname{Null} B_{\mathcal G}(\bm p, \bm R) = \mathcal{T}_b(\bm p, \bm R)$.
    Equivalently, if $\operatorname{rank} B_{\mathcal G}(\bm p, \bm R) = (d+d')|\mathcal{V}| - (d + d' + 1)$.
\end{definition}
Infinitesimal bearing rigidity allows for the local recovery of the true poses up to a similarity transformation.
This is formalized by the following proposition, which follows directly from \cite[Theorem 4]{Michieletto2021TCNS}.
\begin{proposition}
    Let $(\mathcal G, (\bm p, \bm R))$ be an infinitesimally bearing rigid framework. Then, there exists an open neighborhood $\mathcal{U}$ of $(\bm p, \bm R)$ such that if $(\bm q, \bm Q) \in \mathcal{U}$ and $\bm{\beta}_{\mathcal G}(\bm p, \bm R) = \bm{\beta}_{\mathcal G}(\bm q, \bm Q)$ then $(\bm p, \bm R)$ and $(\bm q, \bm Q)$ are related by a similarity transformation.
    \label{prop:inf_bearing_rigidity}
\end{proposition}

\subsection{Angle Rigidity}

Angle rigidity studies the conditions under which robot positions can be reconstructed from angle measurements, up to a similarity transformation.
Angle measurements are independent of the robots' body frames; hence, the robots can  be modeled as $\bm p \in \mathbb{R}^{d|\mathcal{V}|}$.
Some approaches, such as \cite{Chen2022TSP}, incorporate a common reference direction (e.g., gravity), which provides additional information and can relax the requirements on the sensing topology.
In this work, however, we assume no prior knowledge of the robots’ orientations, as this corresponds to the most general setting.

Consider a framework $(\mathcal G, \bm p)$ and let $\mathcal A$ be the set of angle indices induced by $\mathcal E$.
Let $\bm{\alpha}_{\mathcal G} : \mathbb{R}^{d|\mathcal{V}|} \to \mathbb{R}^{|\mathcal{A}|}$ such that $\bm{\alpha}_{\mathcal G}(\bm p) = (\alpha_{ijk})_{(i, j, k) \in \mathcal{A}}$ be the \textit{angle function}, and let the \textit{angle rigidity matrix} $A_{\mathcal G}(\bm p) \in \mathbb{R}^{|\mathcal{A}| \times d|\mathcal V|}$ be the differential of $\bm{\alpha}_{\mathcal G}$.
The latter relates tangent vectors, i.e., $\dot{\bm \alpha}_{\mathcal G} = A_{\mathcal G}(\bm p) \bm{v}$, where $\bm{v} = [v_i]_{i \in \mathcal{V}} \in \mathbb{R}^{d|\mathcal{V}|}$. 
To make the structure of the rigidity matrix explicit, consider $A_{ijk}(\bm p) \in \mathbb{R}^{1 \times d|\mathcal{V}|}$, i.e., the row corresponding to $(i, j, k) \in \mathcal{A}$.
Then,
\begin{equation}
    \dot{\alpha}_{ijk}\!=\!A_{ijk} (\bm p) \bm{v} =  
    \beta_{ik}^\top \frac{P_{ij}}{d_{ij}} (v_j - v_i) + \beta_{ij}^\top \frac{P_{ik}}{d_{ik}} (v_k - v_i).
    \label{eq:angle_rigidity_perturbation}
\end{equation}
In the angle case, the space of infinitesimal motions induced by similarity transformations is
\begin{equation}
\begin{split}
    \mathcal{T}_a(\bm p)\!\coloneqq\!\big\{&\bm{v}\!: v_i \bm\!= \kappa p_i + S(\omega)p_i + v, \\
    &\kappa \in \mathbb{R}, \ \omega \in \mathbb{R}^{d'}, \ v \in \mathbb{R}^d, \ \forall i \in \mathcal{V}\big\}.
\end{split}
\label{eq:angle_inf_trivial}
\end{equation}
It follows that $\operatorname{dim} \mathcal{T}_a(\bm p) = d + d' + 1$ and $\mathcal{T}_a(\bm p) \subseteq \operatorname{Null} A_{\mathcal G}(\bm p)$.
We define infinitesimal angle rigidity as follows.

\begin{definition}
    A framework $(\mathcal G, \bm p)$ in $\mathbb{R}^d$ is said to be infinitesimally angle rigid (IAR) if $\operatorname{Null} A_{\mathcal G}(\bm p) = \mathcal{T}_a(\bm p)$.
    Equivalently, if $\operatorname{rank}A_{\mathcal G}(\bm p) = d|\mathcal{V}| - (d + d' + 1)$.
    \label{def:inf_ang_rigid}
\end{definition}

Infinitesimal angle rigidity enables the local recovery of the true positions up to a similarity transformation.
\begin{proposition}[{\cite[Theorem 2]{Jing2019AUT}}]
    Let $(\mathcal G, \bm p)$ be an infinitesimally angle rigid framework. Then, there exists an open neighborhood $\mathcal{U}$ of $\bm p$ such that if $\bm q \in \mathcal{U}$ and $\bm{\alpha}_{\mathcal G}(\bm p) = \bm{\alpha}_{\mathcal G}(\bm q)$
    then $\bm p$ and $\bm q$ are related by a similarity transformation.
    \label{prop:inf_angle_rigidity}
\end{proposition}

\section{Relation between Angle and Bearing Rigidity}
\label{sec:bearing_vs_angle}

In this section, we present the formal relationship between bearing rigidity and angle rigidity for arbitrary directed graphs in $\mathrm{SE}(d)$.

\begin{definition}
    A framework $(\mathcal G, \bm p)$ in $\mathbb{R}^{d|\mathcal{V}|}$ is said to be degenerate if there exists $i \in \mathcal{V}$, such that $|\mathcal{O}_i| \geq d$ and $\operatorname{dim}\mathcal{B}_i \leq d-1$, where $\mathcal{B}_i \coloneqq \operatorname{span} \{\beta_{ij} : j \in \mathcal{O}_i\} \subseteq \mathbb{R}^d$.
    \label{def:degenerate}
\end{definition}

Definition \ref{def:degenerate} characterizes a framework as degenerate if it includes at least one vertex that has enough outgoing bearings to span $\mathbb{R}^d$, yet those bearings are contained in a proper subspace.
Fig. \ref{fig:degeneracy} illustrates this condition.
It can be checked that the frameworks in Figs. \ref{fig:degenerate_ex} and \ref{fig:degenerate_ex_r3} are infinitesimally bearing rigid but not infinitesimally angle rigid.
On the other hand, the frameworks in Figs. \ref{fig:non_degenerate_ex} and \ref{fig:non_degenerate_ex_r3} are both infinitesimally bearing rigid and infinitesimally angle rigid.
Degeneracy is relevant to our purposes since, for degenerate frameworks, the equivalence between bearing rigidity and angle rigidity presented in Theorem \ref{thm:ibr_iar_equivalence_r3} does not hold. 
However, as shown in Proposition \ref{prop:generic_nondegenerate}, the non-degeneracy condition holds for almost all configurations $\bm p \in \mathbb{R}^{d|\mathcal{V}|}$ (Lebesgue-a.e.).

\begin{figure}[!tb]
    \centering
    \begin{subfigure}{0.49\columnwidth}
        \centering
        \includegraphics[scale=0.9]{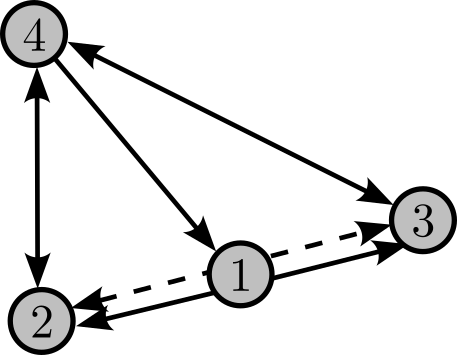}
        \caption{Degenerate.}
        \label{fig:degenerate_ex}
    \end{subfigure}
    \begin{subfigure}{0.49\columnwidth}
        \centering
        \includegraphics[scale=0.9]{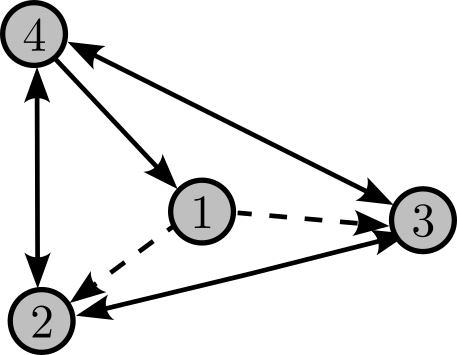}
        \caption{Non-degenerate.}
        \label{fig:non_degenerate_ex}
    \end{subfigure}
    \par\medskip
    \begin{subfigure}{0.49\columnwidth}
        \centering
        \includegraphics[scale=0.9]{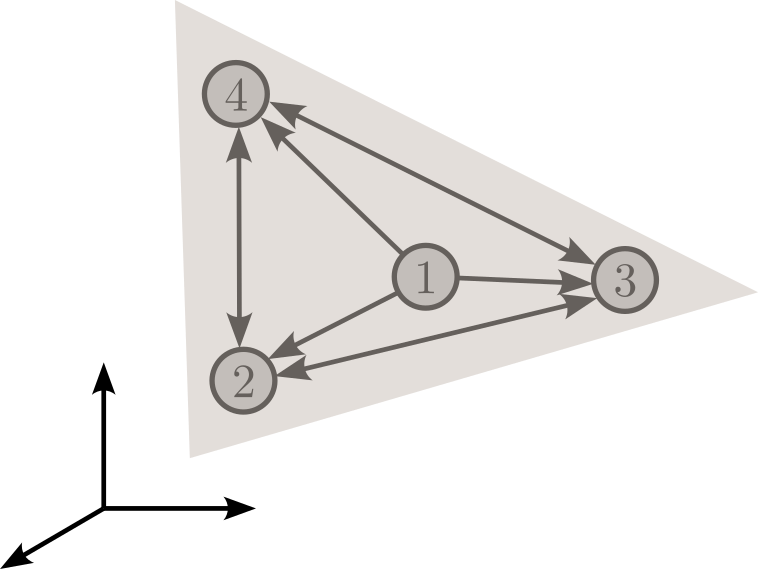}
        \caption{Degenerate.}
        \label{fig:degenerate_ex_r3}
    \end{subfigure}
    \begin{subfigure}{0.49\columnwidth}
        \centering
        \includegraphics[scale=0.9]{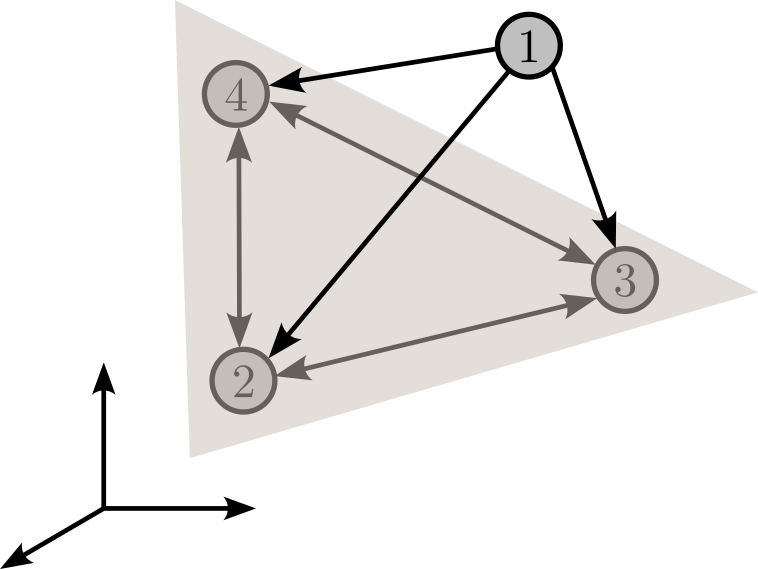}
        \caption{Non-degenerate}
        \label{fig:non_degenerate_ex_r3}
    \end{subfigure}
    \caption{Configuration degeneracy: (a) and (c) degenerate configurations in $\mathbb{R}^2$ and $\mathbb{R}^3$ due to colinearity and coplanarity between $\{1, 2, 3\}$ and $\{1, 2, 3, 4\}$, respectively; (b) and (d) breaking colinearity (coplanarity) removes degeneracy.}
    \label{fig:degeneracy}
\end{figure}

\begin{proposition}
\label{prop:generic_nondegenerate}
Let $\mathcal G=(\mathcal{V},\mathcal E)$ be a directed graph and let $d\ge 2$.
Consider the set of \emph{degenerate} configurations 
\begin{equation*}
    \mathcal{D}(\mathcal{G}) \coloneqq \big\{\bm p \in \mathbb{R}^{d|\mathcal{V}|} : (\mathcal G, \bm p) \ \text{is degenerate}\big\}.
\end{equation*}
Then $\mathcal{D}$ has Lebesgue measure zero.
Consequently, non-degenerate configurations are generic in both $\mathbb{R}^2$ and $\mathbb{R}^3$.
\end{proposition}
\begin{proof}
Fix $i\in\mathcal{V}$ and assume $|\mathcal{O}_i|\ge d$.
From Definition \ref{def:degenerate} and since the linear span does not change under nonzero scalar normalization, it follows that $\mathcal{B}_i = \mathrm{span}\{p_{ij}:j\in\mathcal{O}_i\}$.
Therefore, $\dim \mathcal B_i \le d-1$ is equivalent to saying that the family $\{p_{ij}\}_{j\in\mathcal{O}_i}$ is linearly dependent.
Equivalently, for every $d$-tuple $J=\{j_1,\ldots,j_d\}\subseteq\mathcal{O}_i$, it holds that
\[
\det M_{i, J}(\bm p) = 0, \quad M_{i,J}(\bm p) \coloneqq \big[p_{ij_1} \ \cdots \ p_{ij_d}\big]\in\mathbb{R}^{d\times d}.
\]
Let $f_{i, J}(\bm p) \coloneqq \det M_{i, J}(\bm p)$. It is clear that $f_{i, J}$ is analytic since it is a polynomial in the coordinates of $\bm p$.
Also, $f_{i, J}$ is not identically zero since $f_{i, J}(\bar{\bm p}) = \det(\bm{I}_d)\neq 0$ whenever $\bar{\bm p}$ satisfies $\bar{p}_{ij_k} = e_k$, $k=1,\dots,d$.
It is well-known that the preimage of $0$ under a nonzero analytic function has Lebesgue measure zero.
Finally, by Definition~\ref{def:degenerate},
\[
\mathcal D
 = 
\bigcup_{\substack{i\in\mathcal{V}\\ |\mathcal{O}_i|\ge d}}
\Big\{\bm p: \dim\mathcal B_i\le d-1\Big\}
=
\bigcup_{\substack{i\in\mathcal{V}\\ |\mathcal{O}_i|\ge d}}
\ \ \bigcap_{\substack{J\subseteq\mathcal{O}_i\\ |J|=d}}
f^{-1}_{i,J}(0).
\]
A finite union of measure-zero sets has measure zero, and therefore
$\mathcal D$ has Lebesgue measure zero.
\end{proof}

Now we present the main result of Section \ref{sec:bearing_vs_angle}.
This result shows that angle-based multi-robot coordination is more advantageous than bearing-based coordination, because the set of IAR frameworks strictly contains the set of IBR frameworks.
In practice, using angles allows some robots to dedicate their sensors to useful tasks instead of engaging in mutual observation at all times.

\begin{theorem}
    Let $d \in \{2,3\}$ and let $(\mathcal G, (\bm p, \bm R))$ be a framework in $\mathrm{SE}(d)$ such that $(\mathcal G, \bm p)$ is non-degenerate.
    Then, $(\mathcal G, (\bm p, \bm R))$ is infinitesimally bearing rigid if and only if $(\mathcal G, \bm p)$ is infinitesimally angle rigid and $\operatorname{dim} \mathcal{B}_i \geq d-1$, for all $i \in \mathcal{V}$.
    \label{thm:ibr_iar_equivalence_r3}
\end{theorem}
\begin{proof}
    For convenience, we will focus on the case $d=3$; the case $d=2$ follows analogously.

    \emph{Sufficiency:} We proceed by the contrapositive, which leads us to two cases.
    
    First, assume that $\operatorname{dim} \mathcal{B}_{i^\ast} \leq d - 2 = 1$ for some $i^\ast \in \mathcal{V}$.
    On one hand, if $\mathcal{B}_{i^\ast} = \{\bm 0_d\}$ then $(\mathcal G, (\bm p, \bm R))$ cannot be IBR, because $\bm{\beta}_{\mathcal G}$ would be invariant under changes in $R_{i^\ast}$.
    On the other hand, consider $\operatorname{dim}(\mathcal{B}_{i^{\ast}}) = 1$. Define $\bar{\bm \omega} = [\bar{\omega}_i]_{i \in \mathcal{V}} \in \mathbb{R}^{3|\mathcal{V}|}$  such that $\bar{\omega}_i = 0$ for all $i \in \mathcal{V} \setminus \{i^{\ast}\}$ and $\bar{\omega}_{i^{\ast}} = \beta_{i^\ast j^\ast}$ for some $j^{\ast} \in \mathcal{O}_{i^{\ast}}$. It follows that $S(\bar{\omega}_{i^{\ast}}) \beta_{i^{\ast}j} = 0$ for all $j \in \mathcal{O}_{i^{\ast}}$.
    Let $\bar{\bm v} = 0 \in \mathbb{R}^{3|\mathcal{V}|}$.
    Then $[\bar{\bm v}; \bar{\bm \omega}] \notin \mathcal{T}_b(\bm p, \bm R)$ and $B_{\mathcal G}(\bm p, \bm R) [\bar{\bm v}; \bar{\bm \omega}] = 0$ (see \eqref{eq:bearing_rigidity_perturbation}); hence, the framework is not IBR.
    
    Second, consider $(\mathcal G, \bm p)$ that is not IAR for which $\operatorname{dim}(\mathcal{B}_i) \geq d - 1 = 2$ for each $i \in \mathcal{V}$.
    The first condition means that we can pick $\bar{\bm v} \notin \mathcal{T}_a(\bm p)$ such that $A_{\mathcal G}(\bm p) \bar{\bm v} = 0$.
    Our goal is to show that, irrespective of the choice of $\bm R \in \mathrm{SE}(d)^{|\mathcal{V}|}$, the framework $(\mathcal{G}, (\bm p, \bm R))$ is not IBR.
    To do that, we will construct $\bar{\bm \omega}$ such that $[\bar{\bm v}; \bar{\bm \omega}] \notin \mathcal{T}_b(\bm p, \bm R)$ and $B_{\mathcal G}(\bm p, \bm R)[\bar{\bm v}; \bar{\bm \omega}] = 0$.
    The latter equates to (see \eqref{eq:bearing_rigidity_perturbation})
    \begin{equation}
        S(\beta_{i\nu}) \bar{\omega}_i = - \frac{P_{i\nu}}{d_{i\nu}} (\bar{v}_{\nu} - \bar{v}_i) \quad \text{for all} \quad (i, \nu) \in \mathcal{E}.
        \label{eq:bearing_not_perturbed_ij}
    \end{equation}   
    Hence, for each $i \in \mathcal{V}$ pick $j^{\ast}, k^{\ast} \in \mathcal{O}_i$ such that $\beta_{ij^{\ast}}$ and $\beta_{ik^{\ast}}$ are linearly independent.
    We now prove that there exists $\bar{\omega}_i \in \mathbb{R}^3$ such that \eqref{eq:bearing_not_perturbed_ij} holds for $\nu = j^\ast$ and for $\nu = k^\ast$ which, due to Lemma \ref{lem:linear_system_2} (see Appendix \ref{apx:lemmas}), is equivalent to
    \begin{equation}
    \begin{split}
        \biggr(& S(\beta_{ij^{\ast}}) \frac{P_{ij^{\ast}}}{d_{ij^{\ast}}} (\bar{v}_{j^{\ast}} - \bar{v}_i) - \\  &S(\beta_{ik^{\ast}}) \frac{P_{ik^{\ast}}}{d_{ik^{\ast}}} (\bar{v}_{k^{\ast}} - \bar{v}_i)\biggr)^\top
        (S(\beta_{ij^{\ast}}) \beta_{ik^{\ast}}) = 0.
    \end{split}
    \label{eq:bearing_not_perturbed_ij_equiv}
    \end{equation}
    It can be shown that indeed \eqref{eq:bearing_not_perturbed_ij_equiv} holds by using Binet-Cauchy's identity\footnote{Binet-Cauchy: $(S(x) y)^\top (S(v) w) = (x^\top v)(y^\top w) - (x^\top w)(y^\top v)$ for $x, y, v, w \in \mathbb{R}^3$.}, and the hypothesis $A_{ij^{\ast}k^{\ast}} \bar{\bm{v}} = 0$ (see \eqref{eq:angle_rigidity_perturbation}).
    
    Now, we show that $\bar{\omega}_i$ in fact satisfies \eqref{eq:bearing_not_perturbed_ij} for all $(i, l)$ such that $l \in \mathcal{O}_i$.
    From the non-degeneracy assumption, if $|\mathcal{O}_i| \geq 3$ there exists $l^{\ast} \in \mathcal{O}_i$ such that $\{\beta_{ij^{\ast}}, \beta_{ik^{\ast}}, \beta_{il^{\ast}}\}$ is linearly independent.
    Also,
    \begin{equation*}
        \!\beta_{i \nu}^\top \left(\frac{P_{il^{\ast}}}{d_{il^{\ast}}} (\bar{v}_{l^{\ast}}\!-\!\bar{v}_i) + S(\beta_{il^{\ast}}) \bar{\omega}_i \right)\!=\!0, \ \ \nu\!\in\!\{j^\ast, k^\ast, l^\ast\}.
    \end{equation*}
    When $\nu=l^\ast$, this holds trivially; when $\nu \in \{j^\ast, k^\ast\}$, this follows from $A_{ij^{\ast}l^{\ast}} \bar{\bm{v}} = A_{ik^{\ast}l^{\ast}} \bar{\bm{v}} = 0$ (see \eqref{eq:angle_rigidity_perturbation}) and \eqref{eq:bearing_not_perturbed_ij}. 
    For the remaining $m \in \mathcal{O}_i \setminus \{j^{\ast}, k^{\ast}, l^{\ast}\}$, repeat the process. 
    Then, \eqref{eq:bearing_not_perturbed_ij} follows for all $i \in \mathcal{V}$.
    Finally, since $\bar{\bm v} \notin \mathcal{T}_a$, it follows that $[\bar{\bm v}; \bar{\bm \omega}] \notin \mathcal{T}_b$.

    \emph{Necessity:} we proceed with the following strategy. 
    By assuming that $(\mathcal G, (\bm p, \bm R))$ is not IBR and that $\operatorname{dim}(\mathcal{B}_i) \geq d-1 = 2$ for all $i \in \mathcal{V}$, we show that $(\mathcal G, \bm p)$ is not IAR.
    Consider $[\bar{\bm v}; \bar{\bm \omega}] \notin \mathcal{T}_b(\bm p, \bm R)$ such that $B_{\mathcal G}(\bm p, \bm R) [\bar{\bm v}; \bar{\bm \omega}] = 0$.
    Thus, \eqref{eq:bearing_not_perturbed_ij} holds.
    Also, for $(i, j, k) \in \mathcal{A}$,
    \begin{align*}
        A_{ijk} (\bm p) \bar{\bm v} &= \beta_{ik}^\top \frac{P_{ij}}{d_{ij}} (\bar{v}_j - \bar{v}_i) + \beta_{ij}^\top \frac{P_{ik}}{d_{ik}} (\bar{v}_k - \bar{v}_i)  \\
        &=  - \left( \beta_{ik}^\top S(\beta_{ij}) + \beta_{ij}^\top S(\beta_{ik}) \right) \bar{\omega}_i  = 0.
    \end{align*}
    We conclude that $A_{\mathcal G}(\bm p) \bar{\bm v} = 0$.
    Finally, we show that $\bar{\bm v} \notin \mathcal{T}_a(\bm p)$, for if $\bar{\bm v} \in \mathcal{T}_a(\bm p)$, then from \eqref{eq:angle_inf_trivial} there would exist $\omega \in \mathbb{R}^3$ such that $\bar{v}_j - \bar{v}_i = (\kappa \bm{I}_d + S(\omega))p_{ij}$.
    Then, noting that $P_{ij} S(p_{ij}) = S(p_{ij})$ yields
    \begin{equation}
        P_{ij} (\bar{v}_j - \bar{v}_i) = -S(p_{ij}) \omega, \quad \text{for all} \quad (i, j) \in \mathcal{E}.
        \label{eq:omega_1}
    \end{equation}
    Also, since $B_{\mathcal G}(\bm p, \bm R)[\bar{\bm v}; \bar{\bm \omega}] = 0$, it follows from \eqref{eq:bearing_rigidity_perturbation} that
    \begin{equation}
        P_{ij} (\bar{v}_j - \bar{v}_i) = -S(p_{ij}) \bar{\omega}_i, \quad \text{for all} \quad (i, j) \in \mathcal{E}.
        \label{eq:omega_2}
    \end{equation}
    For all $i \in \mathcal{V}$, \eqref{eq:omega_1} and \eqref{eq:omega_2} imply $S(p_{ij}) \bar{\omega}_i = S(p_{ij}) \omega$ for each $j \in \mathcal{O}_i$.
    Thus, since $\operatorname{dim}(\mathcal{B}_i) \geq 2$, it follows that $\bar{\omega}_i = \omega$.
    This implies $[\bar{\bm v}; \bar{\bm \omega}] \in \mathcal{T}_b$, a contradiction; hence, $\bar{\bm v} \notin \mathcal{T}_a$ and the claim follows.
\end{proof}

\section{Angle-based Estimation and Control}
\label{sec:control}

In this section, we first propose an angle-based localization algorithm and show its local exponential stability under switching topologies, provided that all visited graphs are infinitesimally rigid.
Then, to ensure the convergence of this estimator during a prescribed mission, we develop a decentralized gradient-based controller aimed at maintaining an appropriate level of angle rigidity despite changes in the sensing topology while driving the robots to fulfill the mission.

\subsection{The angle rigidity eigenvalue}
Although Definition \ref{def:inf_ang_rigid} characterizes infinitesimal rigidity, it fails to provide a measure of how rigid a framework is.
To address this, we define the \textit{symmetric angle rigidity matrix} as $\bm A_{\mathcal G}(\bm p) \coloneqq A_{\mathcal G}(\bm p)^\top A_{\mathcal G}(\bm p)$.
It follows that $(\mathcal G, \bm p)$ is IAR if and only if the \textit{angle rigidity eigenvalue} is strictly positive, that is,
\begin{equation}
    \lambda_{d + d' + 2}\left(\bm A_{\mathcal G}(\bm p)\right) > 0, 
    \label{eq:rigidity_eigenvalue}
\end{equation}
where $\lambda_k$ denotes the $k$-th smallest eigenvalue of $\bm A_{\mathcal G}(\bm p)$.
For convenience, we will omit the subscript and denote the rigidity eigenvalue simply by $\lambda$.
We propose using it as an indicator of a framework’s degree of rigidity, as has been done in the context of distance and bearing-based rigidity maintenance control \cite{Zelazo2015IJRR,Schiano2017ICRA}.
The rigidity eigenvalue not only characterizes infinitesimal rigidity, but is also crucial for evaluating both the convergence rate and the robustness to measurement noise of angle-based estimators.
See \cite{Liang2023CDC} for a comprehensive analysis of how the eigenvalues of matrices arising in angle-based localization influence the overall performance.
Therefore, the main objective of rigidity maintenance control is to keep this eigenvalue as large as possible.

\subsection{Angle-based Position Estimation}
\label{sec:anl}

To estimate agents' positions from angle measurements, we adapt the approach for formation control proposed in \cite{Jing2019AUT} to the context of network localization and study its convergence under time-varying topologies.
The goal for each agent is to distributedly obtain $\hat{p}_i \in \mathbb{R}^d$, an estimate of its position with respect to some (unknown) common reference frame.
It is clear that the configuration scale cannot be retrieved from angles alone; hence, we assume that one robot can measure the distance to a neighbor. 
Let $\iota, \kappa \in \mathcal{V}$ such that $d_{\iota \kappa}$ is measured by $\iota$.

Consider a finite family of graphs $\{\mathcal{G}_1, \ldots, \mathcal{G}_M\}$ and a fixed collision-free configuration $\bm p \in \mathbb{R}^{d |\mathcal{V}|}$.
Let $s : \mathbb{R}_{\geq 0} \to \{1, \ldots, M\}$ be a piecewise-constant switching signal.
We propose the following family of potential functions indexed by $m$,
\begin{equation}
    \!\mathcal{L}_m(\hat{\bm p})\!\coloneqq\!\tfrac{\gamma_a}{2} \left\| \bm{\alpha}_{\mathcal{G}_m}\!(\hat{\bm p})\!-\!\bm \alpha_{\mathcal{G}_m}\!(\hat{\bm p}) \right\|^2+\tfrac{\gamma_s}{4}\!\left(\!\hat{d}_{\iota \kappa}^{\, 2}\!-\!{d}_{\iota \kappa}^{\, 2}\!\right)^2,
    \label{eq:anl_cost}
\end{equation}
where $\gamma_a, \gamma_s > 0$ and $\hat{d}_{\iota \kappa} \coloneqq \| \hat{p}_{\kappa} - \hat{p}_{\iota} \|$ denotes the estimated distance.  
The proposed estimator is designed as the negative gradient of \eqref{eq:anl_cost}, $\dot{\hat{\bm p}}(t) = - \nabla_{\hat{\bm p}} \mathcal{L}_{s(t)}(\hat{\bm p}(t))$.
For robot $i$, and omitting the dependence on $m$,
\begin{equation}
    \begin{split}
        \dot{\hat{p}}_i & = \gamma_a \sum_{\{j, k\} \subseteq \mathcal{O}_i} (\hat{\alpha}_{ijk} - \alpha_{ijk}) \left(\hat{\eta}_{ijk} + \hat{\eta}_{ikj}\right) - \\
        & \gamma_a \sum_{\{i, k\} \subseteq \mathcal{O}_j} \left(\hat{\alpha}_{jik} - \alpha_{jik}\right) \hat{\eta}_{jik} - \\
        & \gamma_s (\delta_{i\iota} - \delta_{i\kappa})(\hat{d}_{\iota \kappa}^{\, 2}- d_{\iota \kappa}^{\, 2}) (\hat{p}_{\iota} - \hat{p}_{\kappa}), 
    \end{split}
    \label{eq:anl}
\end{equation}
where $\eta_{ijk} \coloneqq P_{ij} \beta_{ik} / d_{ij}$ and $\delta_{ij}$ is Kronecker's delta.
Here, $\hat{x}$ indicates that variable $x$ is computed using the estimated positions.
To compute \eqref{eq:anl} in a decentralized manner, robot $i$ requires: (a) for each $\{j, k\} \subseteq \mathcal{O}_i$, measure $\alpha_{ijk}$ and receive $\hat{p}_j$ and $\hat{p}_k$; (b) for each $j, k$ such that $\{i, k\} \subseteq \mathcal{O}_j$, receive $(\hat{\alpha}_{jik} - \alpha_{jik}) \hat{\eta}_{jik}$, which can be computed by $j$.

The following Proposition studies the convergence properties of \eqref{eq:anl}.
\begin{proposition}
    Assume that $\bm{p}$ is fixed and that $s(t)$ has finitely many discontinuities on every bounded interval.
    Moreover, assume that each $(\mathcal{G}_m, \bm p), m =1, \ldots, M$ is infinitesimally angle rigid.
    Then, under \eqref{eq:anl}, the desired set $\mathcal{S}(\bm p) \coloneqq \{\bm{q} : \bm{q} = [Qp_i + r]_{i \in \mathcal{V}}, \ (r, Q)\in\mathrm{SE}(d)\}$ is locally exponentially stable.
\end{proposition}
\begin{proof}
    Since $\bm p$ is fixed, we denote $\mathcal{S} = \mathcal{S}(\bm p)$.
    Now, $\mathcal{S}$ is smooth embedded submanifold of $\mathbb{R}^{d|\mathcal{V}|}$, hence it has a tubular neighborhood $\mathcal{U}$ for which there is a smooth orthogonal projection, see \cite{Lee2013ISM}.
    Indeed, let $\Pi :\mathcal{U} \to \mathcal{S}$ be such projection, then for each $\hat{\bm p} \in \mathcal{U}$, it holds that $\bm{\delta} \coloneqq \hat{\bm p} - \Pi(\hat{\bm p}) \in N_{\Pi(\hat{\bm p})} \mathcal{S}$.
    Define 
    \begin{equation*}
        \mathcal{W}(\hat{\bm p}) \coloneqq \frac{1}{2}\operatorname{dist}(\hat{\bm p}, \mathcal{S})^2 = \frac{1}{2}\|\bm{\delta}\|^2
    \end{equation*}
    as the Lyapunov function for the proposed estimator.
    Now, let $\hat{\bm q} \coloneqq \Pi(\hat{\bm p})$, it follows that $\dot{\hat{\bm q}} \in T_{\hat{\bm q}} \mathcal{S}$ and $\bm{\delta} \in N_{\hat{\bm q}} \mathcal{S}$. Thus,
    \begin{equation*}
        \dot{\mathcal{W}}(\hat{\bm p}) 
        = \bm{\delta}^\top \dot{\bm{\delta}} 
        = \bm{\delta}^\top \dot{\hat{\bm p}} = - \bm{\delta}^\top \nabla_{\hat{\bm p}} \mathcal{L}_m(\hat{\bm p}).
    \end{equation*}
    Now, we study the behavior of $\nabla_{\hat{\bm p}} \mathcal{L}_m$ near $\mathcal{S}$.
    Since $\nabla_{\hat{\bm p}} \mathcal{L}_m(\hat{\bm q}) = 0$, its Taylor expansion near $\hat{\bm{q}}$ is
    \begin{equation*}
        \nabla_{\hat{\bm p}} \mathcal{L}_m(\hat{\bm p}) = \nabla_{\hat{\bm p}} \mathcal{L}_m(\hat{\bm q} + \bm{\delta}) = \nabla^2_{\hat{\bm p}} \mathcal{L}_m(\hat{\bm q}) \bm{\delta} + \varepsilon_m(\hat{\bm q}, \bm \delta) \|\bm \delta\|,  
    \end{equation*}
    where $\|\varepsilon_m(\hat{\bm q}, \bm \delta)\| \to 0$ as $\bm \delta \to 0$.\\\\
    On the other hand, for each $\bm{q} \in \mathcal{S}$, the Hessian is
    \begin{equation}
        \nabla^2_{\hat{\bm p}} \mathcal{L}_m(\bm q) = \gamma_a \bm{A}_{\mathcal{G}_m}(\bm q) + \gamma_s \bm{D}_{\iota \kappa}(\bm q),
        \label{eq:hessian}
    \end{equation}
    where $\bm{A}_{\mathcal{G}_m}$ is the symmetric rigidity matrix and $\bm{D}_{\iota \kappa}\coloneqq 2 (e_{\iota}-e_{\kappa})( e_{\iota}-e_{\kappa})^\top \otimes q_{\iota \kappa} q_{\iota \kappa}^\top$, both positive semidefinite.
    From the infinitesimal rigidity assumption of $\mathcal{G}_m$ it follows that 
    $\operatorname{Null}\left(\bm{A}_{\mathcal{G}_m}(\bm q)\right) = T_{\bm q} \mathcal{S} + \operatorname{span}\{\bm{q}\}$, the last term corresponding to scaling motions, which are eliminated by including the term corresponding to $\bm{D}_{\iota \kappa}(\bm q)$.
    Indeed,
    $\operatorname{Null}\left(\bm{A}_{\mathcal{G}_m}\right) \cap \operatorname{Null}\left(\bm{D}_{\iota \kappa}\right) = T_{\bm q} \mathcal{S}$; hence, the Hessian \eqref{eq:hessian} is positive definite on $N_{\bm q} \mathcal{S}$.
    In particular, there is a lower bound $c > 0$ such that for all $m=1,\ldots,M$,
    \begin{equation}
       \bm{v}^\top \nabla^2_{\hat{\bm p}} \mathcal{L}_m(\bm q) \bm{v} \geq c \|\bm v\|^2 \quad \text{for all} \quad \bm{v} \in N_{\bm q} \mathcal{S}.
        \label{eq:hessian_pos_def}
    \end{equation}
    Note that \eqref{eq:hessian_pos_def} holds uniformly in $\bm q$, since the spectrum of both $\bm{A}_{\mathcal{G}_m}$ and $\bm{D}_{\iota \kappa}$ is invariant along $\bm{q} \in \mathcal{S}$.
    Finally, in a sufficiently small neighborhood,
    \begin{equation*}
    \begin{split}
        \dot{\mathcal{W}}(\hat{\bm p}) &= - \bm{\delta}^\top \nabla^2_{\hat{\bm p}} \mathcal{L}_m(\hat{\bm q}) \bm{\delta} - \bm{\delta}^\top \varepsilon_m(\hat{\bm q}, \bm \delta) \|\bm \delta\| \\
        & \leq \left(- c + \|\varepsilon_m(\hat{\bm q}, \bm \delta)\|\right) \|\bm \delta\|^2 \leq - \tilde{c} \|\bm{\delta}\|^2 = -2 \tilde{c} \mathcal{W}(\hat{\bm p})
        \end{split}
    \end{equation*}
    for some $\tilde{c} > 0$.
    Here, we used $-\bm{\delta}^\top \varepsilon_m(\hat{\bm q}, \bm \delta) \leq \|\bm \delta\| \|\varepsilon_m(\hat{\bm q}, \bm \delta)\|$.
    Since $\mathcal{W}$ is continuous and does not depend on $m$, one gets $\mathcal{W}(\hat{\bm p}(t)) \leq e^{-2 \tilde{c} (t-t_0)} \mathcal{W}(\hat{\bm p}(t_0))$,
    and the claim follows.
\end{proof}

\begin{remark}
The use of angle measurements effectively decouples the estimation of the agents’ positions from their orientations.
This is important because, in bearing-based localization, the position and orientation estimators are inherently coupled.
As a result, the propagation of noisy orientation errors not only complicates the analysis of the position estimator but may also compromise its convergence.
\end{remark}

\subsection{Decentralized Rigidity Maintenance Control}
\label{sec:arm}

\subsubsection{Robot model}
We consider a team of robots moving in $\mathbb{R}^d$ with single-integrator dynamics, 
\begin{equation*}
\dot p_i = R_i u^i_i, \qquad \dot R_i = R_i S(\omega^i_i), \qquad i \in \mathcal{V},
\end{equation*}
where $u^i_i$ and $\omega^i_i$ denote the control
inputs, corresponding to linear and angular velocity in the body-frame, respectively.
They are defined as
\begin{equation}
    \begin{split}
        u^i_i &\coloneqq \gamma_r u^i_{i, r} +\gamma_c u^i_{i, c} + \gamma_m u^i_{i, m}, \\ 
        \omega^i_i &\coloneqq \gamma'_r \omega^i_{i, r} + \gamma'_m \omega^i_{i, m}
    \end{split}
    \label{eq:control_action}
\end{equation}
where $\gamma_r,\gamma_c, \gamma_m, \gamma'_r, \gamma'_m > 0$ and each term corresponds to three objectives: rigidity maintenance (Section \ref{sec:rigidity_control}), collision avoidance (Section \ref{sec:collision_control}), and mission completion (Section \ref{sec:mission_control}).

In order to encode sensing constraints into the system's model, let the camera's optical axis be oriented towards the body-frame $x$-axis and represented by $R_i e_1$.
Also, define $\zeta_{ij} \coloneqq \beta_{ij}^\top R e_1 = (\beta^i_{ij})^\top e_1$.
Then, the sensing graph is determined by the edge set
\begin{equation}
    \mathcal{E}(\bm p, \bm R) = \{(i, j) \in \mathcal{V} \times \mathcal{V} : d_{ij} < \rho_r, \ \zeta_{ij} > \rho_f\},
    \label{eq:pose_dependence}
\end{equation}
where $\rho_r > 0$ and $\rho_f \coloneqq \cos(\phi/2)$ with $0 \leq \phi < 2\pi$ are fixed and characterize the cameras' range and field of view, respectively.
Note that, with this model, the angle set $\mathcal A$ also depends on the joint pose $(\bm p, \bm R)$.

\subsubsection{The weighted angle rigidity eigenvalue}

Note that changes in the sensing topology introduce discontinuities in the rigidity eigenvalue \eqref{eq:rigidity_eigenvalue}.
Therefore, due to \eqref{eq:pose_dependence}, $\lambda$ depends non-smoothly on $(\bm p, \bm R)$, which prevents its direct use in gradient-based controllers. To address this, we propose an approximation that \textit{smoothly} accounts for the sensing constraints.
We do this by assigning pose-dependent weights to each angle in $\mathcal{A}$.
These weights are incorporated into the scheme via the \textit{weighted} symmetric angle rigidity matrix 
\begin{equation}
    \widetilde{\bm A}(\bm p, \bm R) \coloneqq A_{\mathcal G}(\bm p)^\top W(\bm p, \bm R) A_{\mathcal G}(\bm p),
    \label{eq:symmetric_rigidity}
\end{equation}
where $\mathcal{G} = \mathcal{G}(\bm p, \bm R)$, and $W \in \mathbb{R}^{|\mathcal A| \times |\mathcal A|}$ is a diagonal matrix that collects all the weights.
The corresponding \textit{weighted} rigidity eigenvalue is denoted by $\widetilde{\lambda}(\bm p, \bm R)$.

Weights are defined as follows.
Let $w_{ijk} \in \mathbb{R}$ be the weight associated with $(i, j, k) \in \mathcal{A}$.
Similar to \cite[Section III.D]{Schiano2017ICRA}, we employ the family  $\sigma : \mathbb{R} \to [0, 1]$ of functions parameterized by $a, b \in \mathbb{R}$ defined as
\begin{equation}
    \sigma(x; a, b) = 
    \begin{cases}
        0, & x < a \\
        \frac{1}{2} \left(1 - \cos\left(\pi \frac{x - a}{b - a}\right)\right), & a \leq x \leq b \\
        1, & b < x 
    \end{cases}.
    \label{eq:sigma}
\end{equation}
Hence, $\sigma(x; a, b)$ transitions (with a continuous derivative) from $0$ to $1$ as $x$ varies from $a$ to $b$.
Let $\sigma_r(x) \coloneqq \sigma(x; a_r, b_r)$ with $0 \leq a_r < b_r \leq \rho_r$ and define the weight $w_{r_{ij}} \coloneqq 1 - \sigma_r(d_{ij})$ associated with edge $(i, j)$, which accounts for the limited sensing range.
It follows that $w_{r_{ij}} = 1$ when $d_{ij} < a_r$ and smoothly approaches $0$ as $d_{ij}$ increases to $b_r$.
Likewise, $\sigma_f(x) \coloneqq \sigma(x; a_f, b_f)$ with $\rho_f \leq a_f < b_f \leq 1$ and $w_{f_{ij}} \coloneqq \sigma_f(\zeta_{ij})$, which considers the limited field of view.
Hence, $w_{f_{ij}} = 1$ when $\zeta_{ij} > b_f$ and smoothly approaches $0$ as $\zeta_{ij}$ decreases to $a_f$.
Analogously, define $w_{r_{ik}}$ and $w_{f_{ik}}$ for edge $(i, k)$.
Finally, we set
\begin{equation}
    w_{ijk} = d_{ij} d_{ik} w_{r_{ij}} w_{f_{ij}} w_{r_{ik}} w_{f_{ik}}.
    \label{eq:weights}
\end{equation}
Thus, $w_{ijk} >0$ if and only if agents $j$ and $k$ are both captured by $i$'s camera, and $w_{ijk}$ smoothly approaches $0$ as any of them moves outside of it.
The reason to include the factor $d_{ij} d_{ik}$ is to make $\widetilde{\bm A}$ (and therefore $\widetilde{\lambda}$) scale invariant.
By checking \eqref{eq:angle_rigidity_perturbation}, it follows that the symmetric rigidity matrix $\bm A_{\mathcal G}$ (and therefore $\lambda$) depends inversely on the square of the configuration scale.
This implies that following the gradient of the rigidity eigenvalue drives the robots to collapse into a single point, which is undesirable.

\begin{proposition}
    If all weights are defined by \eqref{eq:weights}, then in the open set where no robot collisions occur, $\widetilde{\bm A}$ \eqref{eq:symmetric_rigidity} is continuously differentiable with respect to $(\bm p, \bm R)$.
    \label{prop:symmetric_rigidity}
\end{proposition}
\begin{proof}
    First, $\widetilde{\bm A}(\bm p, \bm R) = \sum_{(i, j, k) \in \mathcal{A}} w_{ijk} A_{ijk}^\top A_{ijk}$.
    The claim follows since all quantities involved are continuously differentiable with respect to $(\bm p, \bm R)$ under the non-collision assumption, and from the fact that $w_{ijk}$ smoothly approaches $0$ as $(i, j, k)$ is removed from $\mathcal{A}$.
\end{proof}

\subsubsection{Rigidity control action}
\label{sec:rigidity_control}
As discussed, angle rigidity maintenance aims to prevent the rigidity eigenvalue from approaching zero.
To this end, we propose the minimization of the following potential
\begin{equation}
    \mathcal{R}(\bm p, \bm R) \coloneqq - \log\left(\widetilde{\lambda}(\bm p, \bm R)\right).
    \label{eq:arm_cost}
\end{equation}
It is important to emphasize that the objective of angle rigidity maintenance fundamentally differs from angle-based formation or shape control. The potential \eqref{eq:arm_cost} does not aim to enforce convergence to a prescribed set of angles or to maintain a fixed geometric shape. Instead, its sole purpose is to prevent the system from evolving toward configurations where the angle rigidity matrix loses rank.
As a result, the angle set (and their values) is not required to remain constant along the trajectory, nor is the formation constrained to move rigidly as a whole.

Now, we obtain expressions for the body-frame rigidity control actions for agent $i$.
Let $\widetilde{\bm \nu} = [\tilde{\nu}_i]_{i \in \mathcal{V}}$ be a unit-norm eigenvector associated with $\widetilde{\lambda}$.
First, we express the weighted rigidity eigenvalue as a sum over all angles,   
\begin{equation}
    \widetilde{\lambda}(\bm p, \bm R) = \widetilde{\bm \nu}^\top \widetilde{\bm A} \widetilde{\bm \nu} = \sum_{(i, j, k) \in \mathcal{A}} w_{ijk} \left(A_{ijk} \widetilde{\bm \nu}\right)^2,
    \label{eq:re}
\end{equation}
with $A_{ijk} \widetilde{\bm \nu}$ as in \eqref{eq:angle_rigidity_perturbation}.
The gradient of \eqref{eq:re} with respect to $(p_i, R_i)$ is
\begin{equation}
    \begin{split}
        \nabla_{p_i} \widetilde{\lambda} &= \sum_{(\iota, \kappa, \ell) \in \mathcal{A}} \biggr(\!\left(A_{\iota\kappa\ell} \widetilde{\bm \nu}\right)^2 \nabla_{p_i} w_{\iota\kappa\ell} \ + \\
        & \quad 2 w_{\iota\kappa\ell} \left(A_{\iota\kappa\ell} \widetilde{\bm \nu}\right) \tfrac{\partial A_{\iota\kappa\ell}}{\partial p_i} \widetilde{\bm \nu} \biggr), 
        \\
        \nabla_{R_i} \widetilde{\lambda} &= \sum_{(\iota, \kappa, \ell) \in \mathcal{A}} \left(A_{i\kappa \ell} \widetilde{\bm \nu}\right)^2 \nabla_{R_i} w_{i\kappa \ell}.
    \end{split}
    \label{eq:re_grad}
\end{equation}

\begin{remark}
To obtain \eqref{eq:re_grad}, we used $\frac{\partial \widetilde{\lambda}}{\partial \xi} = \widetilde{\bm \nu}^\top \frac{\partial \widetilde{\bm A}}{\partial \xi}\widetilde{\bm \nu}$, for each parameter $\xi \in \mathbb{R}$.
This expression assumes that the rigidity eigenvalue has multiplicity one since, otherwise, $\widetilde{\bm \nu}$ is not unique.
Potential workarounds when $\widetilde{\lambda}$ is not simple are discussed in \cite[Section III.D]{Schiano2017ICRA} and \cite[Section IV.B]{Presenza2025CSL}.
However, decentralized and scalable solutions for this problem are still a subject of research.
\end{remark}

Finally, since terms in \eqref{eq:re_grad} vanish when $i \notin \{\iota, \kappa, \ell\}$,
\begin{align}
    \begin{split}
        u^i_{i, r} & \coloneqq - R_i^\top \nabla_{p_i} \mathcal{R}
        = \widetilde{\lambda}^{-1} R_i^\top \nabla_{p_i} \widetilde{\lambda}
        \\
        &= \widetilde{\lambda}^{-1} \sum_{\{j, k\} \subseteq \mathcal{O}_i} \biggr(\left(A_{ijk} \widetilde{\bm \nu}\right)^2 \left(R_i^\top \nabla_{p_i} w_{ijk}\right) \\
        &\qquad + 2 w_{ijk} \left(A_{ijk} \widetilde{\bm \nu}\right) \left(R_i^\top \tfrac{\partial A_{ijk}}{\partial p_i} \widetilde{\bm \nu}\right)\biggr)
        \\
        & + \widetilde{\lambda}^{-1} \sum_{\{i, k\} \subseteq \mathcal{O}_j} \biggr(\left(A_{jik} \widetilde{\bm \nu}\right)^2 R_{ij} \left(R_j^\top \nabla_{p_i} w_{jik}\right) \\
        &\qquad + 2 w_{jik} \left(A_{jik} \widetilde{\bm \nu}\right) R_{ij} \left(R_j^\top \tfrac{\partial A_{jik}}{\partial p_i} \widetilde{\bm \nu}\right)\biggr).
    \end{split}
    \label{eq:re_grad_local_u}
    \\
    \begin{split}
        \omega^i_{i, r} & \coloneqq - S^{-1}\left(R_i^\top \nabla_{R_i} \mathcal{R}\right)
        = \widetilde{\lambda}^{-1} S^{-1} \left(R_i^\top \nabla_{R_i} \widetilde{\lambda}\right) \\
        &= \widetilde{\lambda}^{-1} \sum_{\{j, k\} \subseteq \mathcal{O}_i} \left(A_{ijk} \widetilde{\bm \nu}\right)^2 S^{-1}\left(R_i^\top \nabla_{R_i} w_{ijk}\right) 
    \end{split}
    \label{eq:re_grad_local_w}
\end{align}

\begin{remark}
    Expression \eqref{eq:re_grad_local_u}  contains absolute orientations $R_i, R_j$. However, they are not needed, as the corresponding terms can be computed from body-frame variables; see \eqref{eq:grad_w_1}--\eqref{eq:grad_A_2}.
\end{remark}

\begin{remark}
    Expression \eqref{eq:re_grad_local_u} involves the relative rotation $R_{ij} \coloneqq R_i^\top R_j$, assumed to be available to $j$ if $(j, i) \in \mathcal{E}$.
    This term arises from the need to express gradient contributions in the body reference frames.
    The control itself does not require orientations as part of the state,
     nor are additional constraints imposed to the sensing graph beyond those already required for angle rigidity.
\end{remark}

\begin{remark}
    The rigidity eigenvalue $\widetilde{\lambda}$ and the body-frame eigenvector $\widetilde{\bm \nu}$ are global variables.
    As is common in the literature, we assume that each robot locally estimates $\widetilde{\lambda}$ and $\tilde{\nu}^i_i \coloneqq R_i^\top \tilde{\nu}_i$ by adapting the approach in \cite{Zelazo2015IJRR,Sun2015CCC} to the angle case.
\end{remark}

In what follows, we provide formulas to compute \eqref{eq:re_grad_local_u} and \eqref{eq:re_grad_local_w} from onboard measurements and communications between neighbors.
This requires, for each robot $i$: (a) for each $\{j, k\} \subseteq \mathcal{O}_i$, to measure $\beta^i_{ij}$, $\beta^i_{ik}$, $R_{ij}$, and $R_{ik}$, and to receive $\hat{p}_j$, $\hat{p}_k$, $\tilde{\nu}_j^j$, and $\tilde{\nu}_k^k$; (b) for each $j, k$ such that $\{i, k\} \subseteq \mathcal{O}_j$, to receive the corresponding term in the second summation of \eqref{eq:re_grad_local_u}, which can be computed by $j$.

First, weight $w_{ijk}$ can be computed from the estimated values $d_{ij}$ and $d_{ik}$, and from $\zeta_{ij}$ and $\zeta_{ik}$, obtained from body-frame bearings.
Also,
\begin{equation}
    A_{ijk} \widetilde{\bm{\nu}} =\!(\eta^i_{ijk})^\top(R_{ij} \tilde{\nu}^j_j - \tilde{\nu}^i_i) + (\eta^i_{ikj})^\top(R_{ik} \tilde{\nu}^k_k - \tilde{\nu}^i_i)
    \label{eq:rigidity_control_matrix}
\end{equation}
where $\eta^i_{ijk} \coloneqq P^i_{ij} \beta^i_{ik} / d_{ij}$ and $\eta^i_{ikj} \coloneqq P^i_{ik} \beta^i_{ij} / d_{ik}$.

Second, we specify the gradients of the weights expressed in the body frame.
Let $\sigma'_r$ and $\sigma'_f$ be the derivatives of $\sigma_r$ and $\sigma_f$, respectively.
Then,
\begin{align}
    \begin{split}
    & R_i^\top \nabla_{p_i} w_{ijk} = w_{ijk} \Biggr(\left(\tfrac{\sigma'_r(d_{ij})}{w_{r_{ij}}} - \tfrac{1}{d_{ij}}\right) \beta^i_{ij} + \\
    & \left(\!\tfrac{\sigma'_r(d_{ik})}{w_{r_{ik}}}\!-\!\tfrac{1}{d_{ik}}\right) \beta^i_{ik}\!-\!\tfrac{\sigma'_f \left(\zeta_{ij} \right)}{w_{f_{ij}} d_{ij}} P^i_{ij} e_1\!-\!\tfrac{\sigma'_f(\zeta_{ik})}{w_{f_{ik}} d_{ik}} P^i_{ik} e_1\!\Biggr),
    \end{split}
    \label{eq:grad_w_1}
    \\
    \begin{split}
    & R_i^\top \nabla_{R_i} w_{ijk} =\\
    & \quad w_{ijk} S\left(S(e_1) \left( \tfrac{\sigma'_f\left(\zeta_{ij}\right)}{w_{f_{ij}}} \beta^i_{ij} + \tfrac{\sigma'_f(\zeta_{ik})}{w_{f_{ik}}} \beta^i_{ik} \right) \right), 
    \end{split}
    \label{eq:grad_w_2}
    \\
    \begin{split}
    & R_j^\top \nabla_{p_i} w_{jik} = \\
    & \quad w_{jik} \left(\left(\tfrac{1}{d_{ji}} - \tfrac{\sigma'_r(d_{ji})}{w_{r_{ji}}}\right) \beta^j_{ji} + \tfrac{\sigma'_f \left(\zeta_{ji} \right)}{w_{f_{ji}} d_{ji}} P^j_{ji} e_1\right).
    \end{split}
    \label{eq:grad_w_3}
\end{align}
See Appendix \ref{apx:rigidity_control} for the derivation of the gradients on $\mathrm{SO}(d)$ used in \eqref{eq:grad_w_2}.
It is worth noting that all denominators in \eqref{eq:grad_w_1}--\eqref{eq:grad_w_3} cancel with $w_{ijk}, w_{jik}$, ensuring that all the gradients are well defined. Equivalent expressions without denominators can be derived, but they are more cumbersome.
Additionally, \eqref{eq:grad_w_1}--\eqref{eq:grad_w_2} and \eqref{eq:grad_w_3} can be computed by $i$ and $j$, respectively, using locally available quantities.

Third, we specify the derivatives of $A_{ijk} \widetilde{\bm \nu}$ and $A_{jik} \widetilde{\bm \nu}$.
Define, for each $i, j, k$, the matrices
\begin{align*}
    D^i_{ijk} &\coloneqq \frac{P^i_{ij} \beta^i_{ik} (\beta^i_{ij})^\top + \beta^i_{ij} (\beta^i_{ik})^\top P^i_{ij} + \left((\beta^i_{ij})^\top \beta^i_{ik}\right) P^i_{ij}}{d_{ij}^2}, \\ E^i_{ijk} &\coloneqq \frac{P^i_{ij} P^i_{ik}}{d_{ij} d_{ik}}.
\end{align*}
Then,
\begin{align}
\begin{split}
    R_i^\top \frac{\partial A_{ijk}}{\partial p_i} \widetilde{\bm \nu} & = (D^i_{ijk} - E^i_{ikj}) (R_{ij} \tilde{\nu}^j_j - \tilde{\nu}^i_i) \\
    & + (D^i_{ikj} - E^i_{ijk}) (R_{ik} \tilde{\nu}^k_k - \tilde{\nu}^i_i)
\end{split}
\label{eq:grad_A_1}
\\
\begin{split}    
    R_j^\top \frac{\partial A_{jik}}{\partial p_i} \widetilde{\bm \nu} & = -D^j_{jik} (R_{ji} \tilde{\nu}^i_i - \tilde{\nu}^j_j) \\
    &+ E^j_{jik} (R_{jk} \tilde{\nu}^k_k - \tilde{\nu}^j_j).
\end{split}
\label{eq:grad_A_2}
\end{align}

\subsection{Collision avoidance}
\label{sec:collision_control}

It can be observed that the rigidity maintenance control action depends inversely on the distances from each robot $i$ to its neighbors; see \eqref{eq:rigidity_control_matrix}, \eqref{eq:grad_A_1}, and \eqref{eq:grad_A_2}.
To ensure the well-posedness of the control action, we include a collision avoidance mechanism.
Consider the potential function
\begin{equation}
    \mathcal{C}(\bm p) \coloneqq \sum_{i \in \mathcal{V}} \sum_{j \in \mathcal{O}_i}  \left((d_{ij} - \rho_r) / d_{ij}\right)^2.
    \label{eq:collision_cost}
\end{equation}
Each term of this cost grows unbounded as $d_{ij} \to 0$ and vanishes (with a vanishing derivative) as $d_{ij} \to \rho_r$.
The body-frame collision avoidance action is
\begin{equation}
    \begin{split}
        u^i_{i, c} & \coloneqq - R_i^\top \nabla_{p_i} \mathcal{C} = 2 \rho_r \sum_{j \in \mathcal{O}_{i}} \frac{d_{ij} - \rho_r}{d_{ij}^3} \beta^i_{ij} \\
        &- 2 \rho_r \sum_{i \in \mathcal{O}_j} \frac{d_{ij} - \rho_r}{d_{ij}^3} R_{ij} \beta^j_{ji},
    \end{split}
    \label{eq:collision_action}
\end{equation}
which requires receiving body-frame bearings for each $j$ for which $i \in \mathcal{O}_j$.

\section{Validation}
\label{sec:validation}

\subsection{Cooperative mission}
\label{sec:mission_control}

To validate the proposed approach, a cooperative mission involving multiple target tracking is proposed. 
Each robot $i$ is assigned to track a (potentially distinct) target $\tau_i$.
This is done because, in order to fulfill this objective, robots move and orient their cameras along distinct, time-varying directions.
As a result, changes in the sensing topology are generated, challenging the rigidity condition and testing the capabilities of both the estimator and the controller.
These moving targets may be either physical objects in motion or the output of an algorithm that controls the cameras' motion for environmental exploration and related tasks.

Hence, we employ \eqref{eq:sigma} to define the potential
\begin{equation}
    \mathcal{M}(\bm p, \bm R) \coloneqq \sum_{i \in \mathcal{V}} \mu_r \sigma_r(d_{i\tau_i}) + \mu_f (1 - \sigma_f(\zeta_{i\tau_i})),
    \label{eq:target_cost}
\end{equation}
with $\mu_r, \mu_f > 0$.
For this application, a suitable choice of constants is $\mu_r = 1.5$, $\mu_f = 1.0$, $\sigma_r(x) \coloneqq \sigma(x; 0, \rho_t)$, $\sigma_f(x) \coloneqq \sigma(x; -1, 1)$, where $\rho_t$ can be interpreted as a tracking radius.
The (negative) gradient flow of \eqref{eq:target_cost} attracts the robot $i$ and orients its camera toward a designated moving target $\tau_i$ if $d_{i \tau_i} < \rho_t$. 
The mission-related control action for robot $i$ results in
\begin{equation}
\begin{split}
    u^i_{i, m} & \coloneqq\!- R_i^\top \nabla_{p_i} \mathcal{M} = \mu_r \sigma'_r(d_{i \tau_i}) \beta^i_{i \tau_i} - \mu_f \sigma'_f(\zeta_{i \tau_i}) \tfrac{P^i_{i \tau_i}}{d_{i \tau_i}} e_1 \\
    \omega^i_{i, m} & \coloneqq\!- S^{-1} \left(R_i^\top \nabla_{R_i} \mathcal{M}\right) = \mu_f \sigma'_f(\zeta_{i \tau_i}) S(e_1) \beta^i_{i \tau_i}, 
\end{split}
\label{eq:target_control}
\end{equation}
which requires only the body-frame bearing.
We assume that the distance to the target $d_{i\tau_i}$ is provided by the underlying tracking algorithm, either by specifying it as a setpoint or by estimating it from the camera images. The implementation of this procedure is beyond the scope of this work.
Also, the target tracking control action depends inversely on $d_{i \tau_i}$; therefore, a term corresponding to robot-target collision should be added to \eqref{eq:collision_cost} and \eqref{eq:collision_action}.

\subsection{Simulation Results}

To validate the proposed control scheme, we performed simulations aimed at evaluating its performance and robustness under challenging conditions.
We employed $N = 5$ robots in $\mathbb{R}^3$.
At $t=\SI{0}{\second}$, the robots' positions were randomly generated within a cube of $(40 \si{\meter})^3$ and centered at $[\SI{80}{\meter}; \SI{20}{\meter}; \SI{20}{\meter}]$, see Fig. \ref{fig:snapshots}.
The cameras' axes were oriented toward the barycenter of all positions to obtain a sufficiently dense sensing graph at initialization.
The cameras' parameters were set to $\rho_r = \SI{30}{\meter}$ and $\rho_f = 0.5$ (corresponding to a $\phi = \SI{120}{\degree}$ field-of-view angle).
 The weights were defined using the values $a_r = 0.8 \rho_r$, $b_r = \rho_r$, $a_f = \rho_f$, and $b_f = 1.4 \rho_f$.
The gains of the controller (see \eqref{eq:control_action} and \eqref{eq:target_control}) were $\gamma_r = 5$, $\gamma_c = 1 $, $\gamma_m = 25$, $\gamma'_r = 0.5$, $\gamma'_m = 2.5$, $\mu_r = 6$, and $\mu_f = 1$.
For this setup, $3$ targets were used, which were tracked by robots $1$, $\{2, 4\}$ and $\{3, 5\}$, respectively; hence $\tau_2 = \tau_4$ and $\tau_3 = \tau_5$.
They followed predefined trajectories within the cube $[\SI{0}{\meter}, \SI{100}{\meter}]^3$.
Example videos can be found in \cite{SimuVideo1}.

It can be seen that at $t=\SI{0}{\second}$ all robots are engaged in mutual sensing, ensuring infinitesimal angle rigidity.
Once the system is deployed, each robot $i$ immediately begins to move closer to and aim at its assigned target $\tau_i$. Because of asymmetries in the spatial configuration, some robots will be able to point towards their targets sooner than others, thereby losing visual contact with the rest of the team. The remaining robots must then focus on preserving the necessary bearing measurements (and thus the angles) required to maintain angle rigidity, since the relative importance of $\mathcal R$ eventually exceeds $\mathcal M$.
In this example (Fig. \ref{fig:snapshots}), three robots (the maximum possible for $N=5$) lose visual contact with the team and successfully focus on tracking.  The remaining two will adjust their positions and camera orientations in order to maintain angle rigidity throughout the mission.

Fig. \ref{fig:performance_metrics} presents performance metrics. Fig. \ref{fig:eigenvalues} shows the evolution of the angle rigidity eigenvalue, including both weighted and unweighted versions.
The latter was scaled using the distances $\{d_{ij} d_{ik} : (i, j, k) \in \mathcal{A}\}$ for appropriate comparison.
It can be seen how the rigidity eigenvalue $\lambda_8$ was kept positive at all times, ensuring the system's infinitesimal angle rigidity.
In addition, $\widetilde{\lambda}_8 \leq \lambda_8$ and, although they start close to each other, their difference increases rapidly after $t \approx \SI{40}{\second}$, indicating that the robot configuration induces weights that are positive but small.
This lets the controller know that some angle measurements are close to being lost.
It is interesting to notice that $\lambda_8$ (which is the one that determines the performance of the angle-based estimators) remains approximately constant, despite the number of edges (bearing measurements) decreasing over time as the robots disperse.
This demonstrates the controller’s ability to compensate for the loss of bearing measurements by optimizing the robots’ geometric configuration.

Fig. \ref{fig:targets} displays how the targets are visually tracked by the robots as the mission progresses.
It shows the angles between the camera axes $R e_1$ and the bearing vectors $\beta_{i \tau_i}$, i.e., $\arccos(\zeta_{i\tau_i})$.
For target $\tau_i$ to be captured by robot $i$'s camera, the angle $\arccos(\zeta_{i\tau_i})$ must be smaller than half of the field of view (\SI{60}{\degree}). Moreover, these angles decrease as the robot better aims at the target.
Here, the lines that remain below \SI{60}{\degree} correspond to robots $1$, $4$, and $5$; therefore, each target is successfully tracked by one agent.
Figure~\ref{fig:localization} presents the position estimation error obtained by applying the localization scheme described in Section~\ref{sec:anl}, where $\iota = 0$, $\kappa = 1$, $\gamma_a = 1000$ and $\gamma_s = 0.05$.
For each robot $i$, the initial estimate $\hat{p}_i$ was generated by sampling from a normal distribution with mean $p_i$ and standard deviation $\SI{2}{\meter}$.
The errors show a sustained decrease as the mission progresses and remain below $\approx \SI{0.5}{\meter}$, demonstrating the effectiveness of the angle-based estimator despite time-varying conditions.

A second configuration was tested, using $N=8$ robots and a single target, Fig. \ref{fig:snapshots_b}. In this case, the target moves rapidly along a trajectory contained in the surface of a sphere that encompasses all the robots, inducing fast camera motions and frequent topology changes. 
Performance metrics are shown in Fig. \ref{fig:performance_metrics_2} and example videos in \cite{SimuVideo2}.
As intended, the robots manage to maintain rigidity while most robots points toward the target.
Interestingly, as the mission progresses, the controller drives the system towards a spatial configuration in which loss of visual contact becomes less frequent.

\begin{figure}
    \centering
    \begin{subfigure}{0.49\columnwidth}
        \includegraphics[width=\columnwidth, trim=0mm 0mm 50mm 18mm, clip]{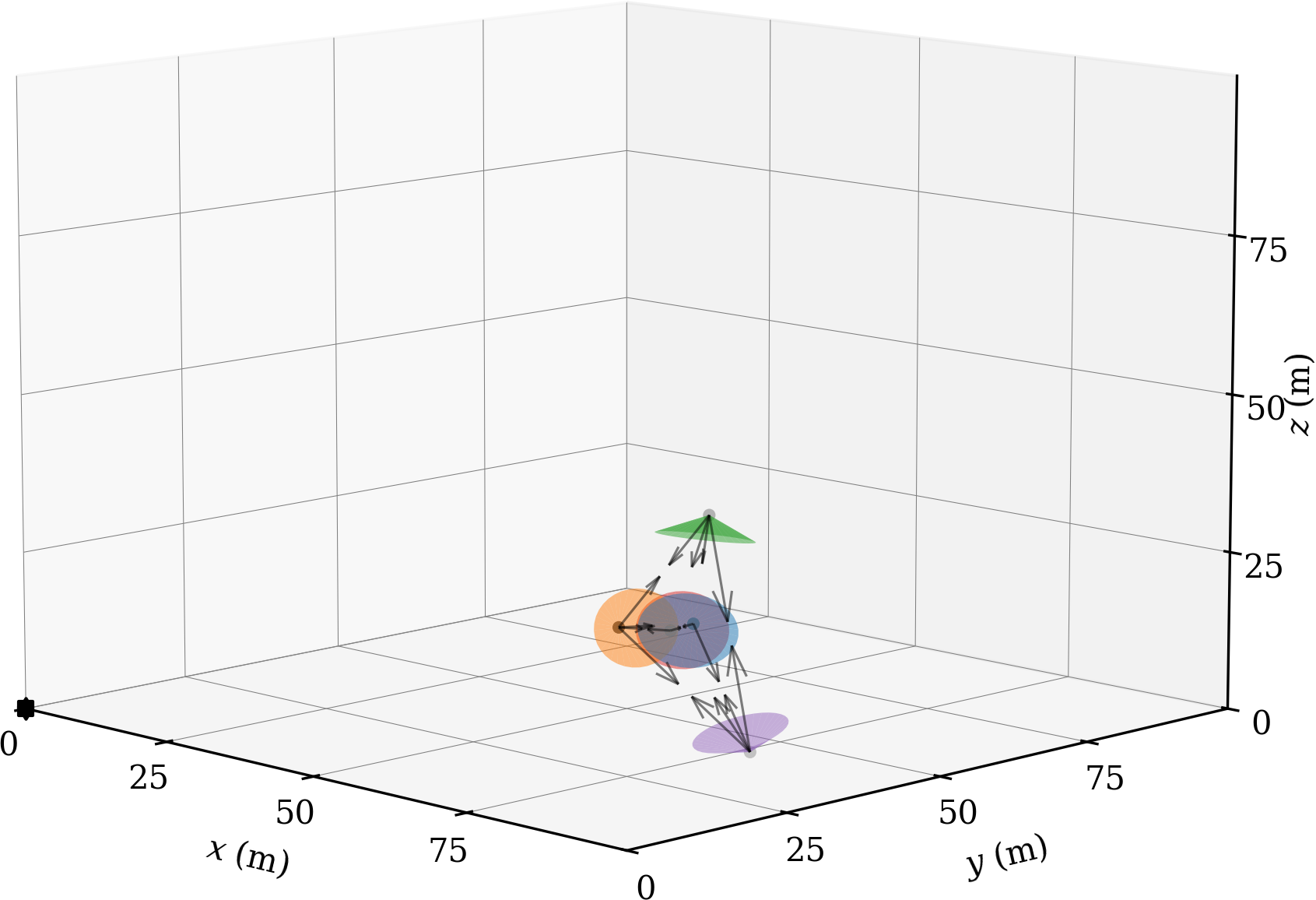}
        \caption{$t = \SI{0}{\second}$.}
    \end{subfigure}
    \begin{subfigure}{0.49\columnwidth}
        \includegraphics[width=\columnwidth, trim=5mm 3mm 45mm 15mm, clip]{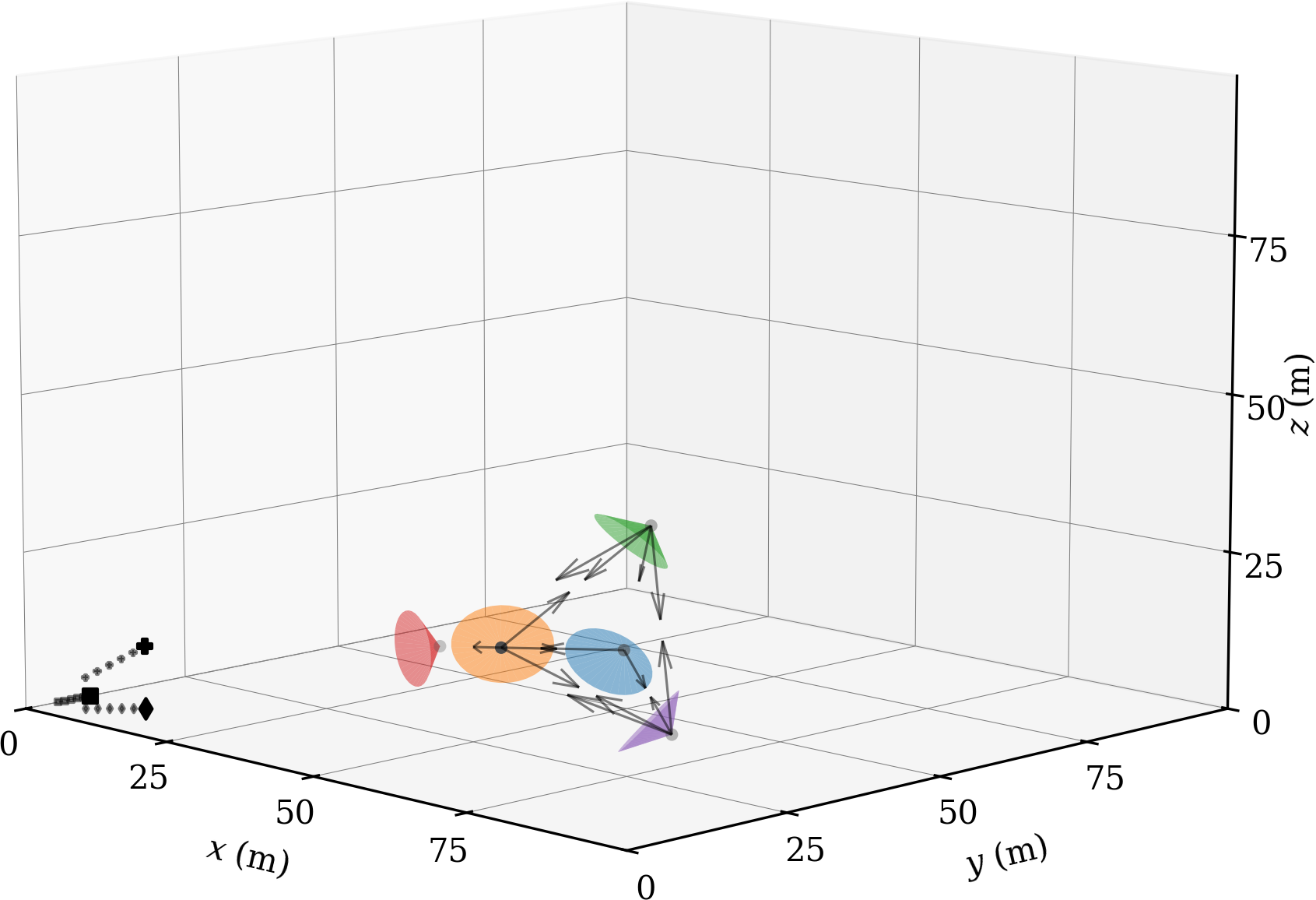}
        \caption{$t = \SI{10}{\second}$.}
    \end{subfigure}
    \begin{subfigure}{0.49\columnwidth}
        \includegraphics[width=\columnwidth, trim=35mm 15mm 15mm 3mm, clip]{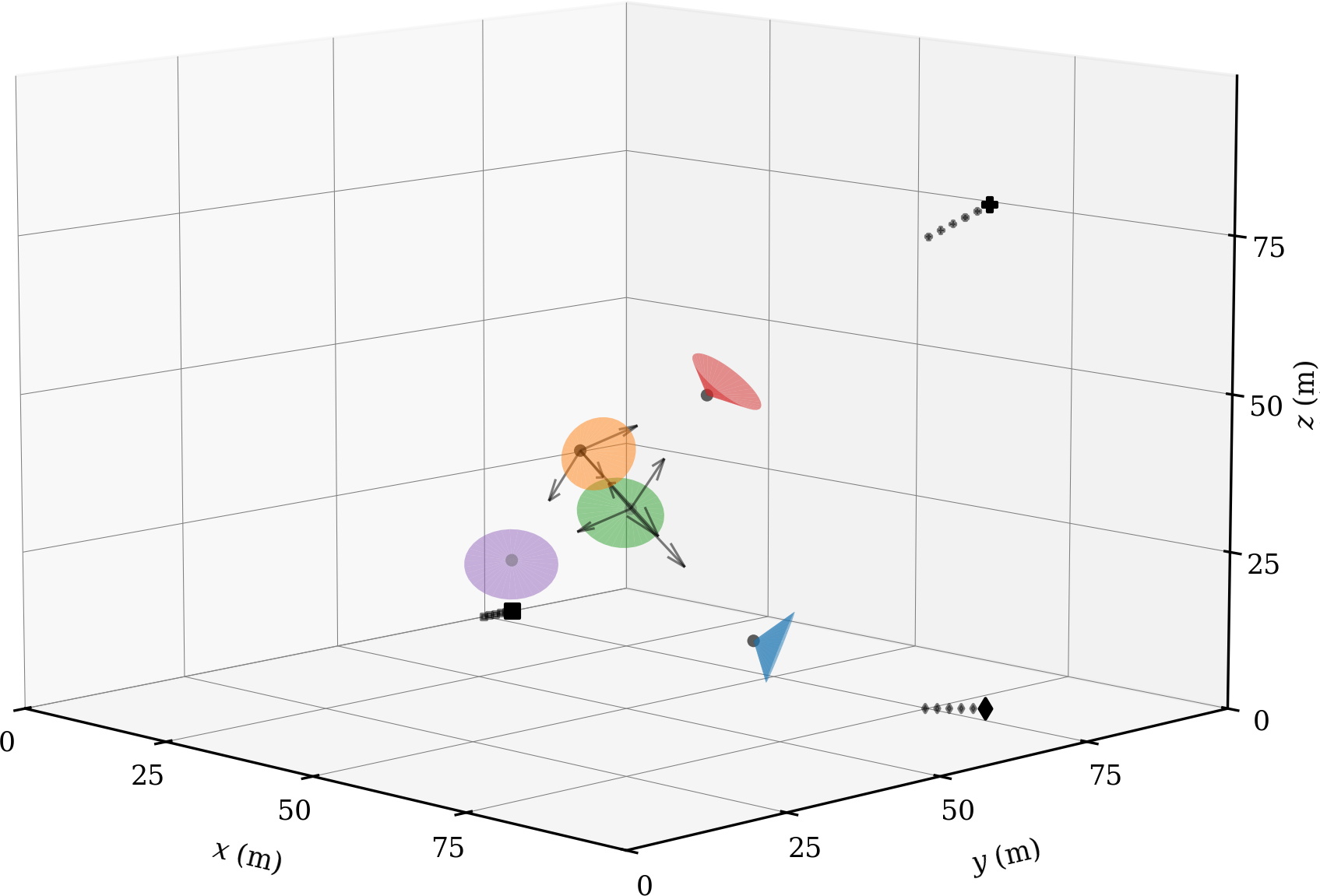}
        \caption{$t = \SI{80}{\second}$.}
    \end{subfigure}
    \begin{subfigure}{0.49\columnwidth}
        \includegraphics[width=\columnwidth, trim=30mm 18mm 20mm 0mm, clip]{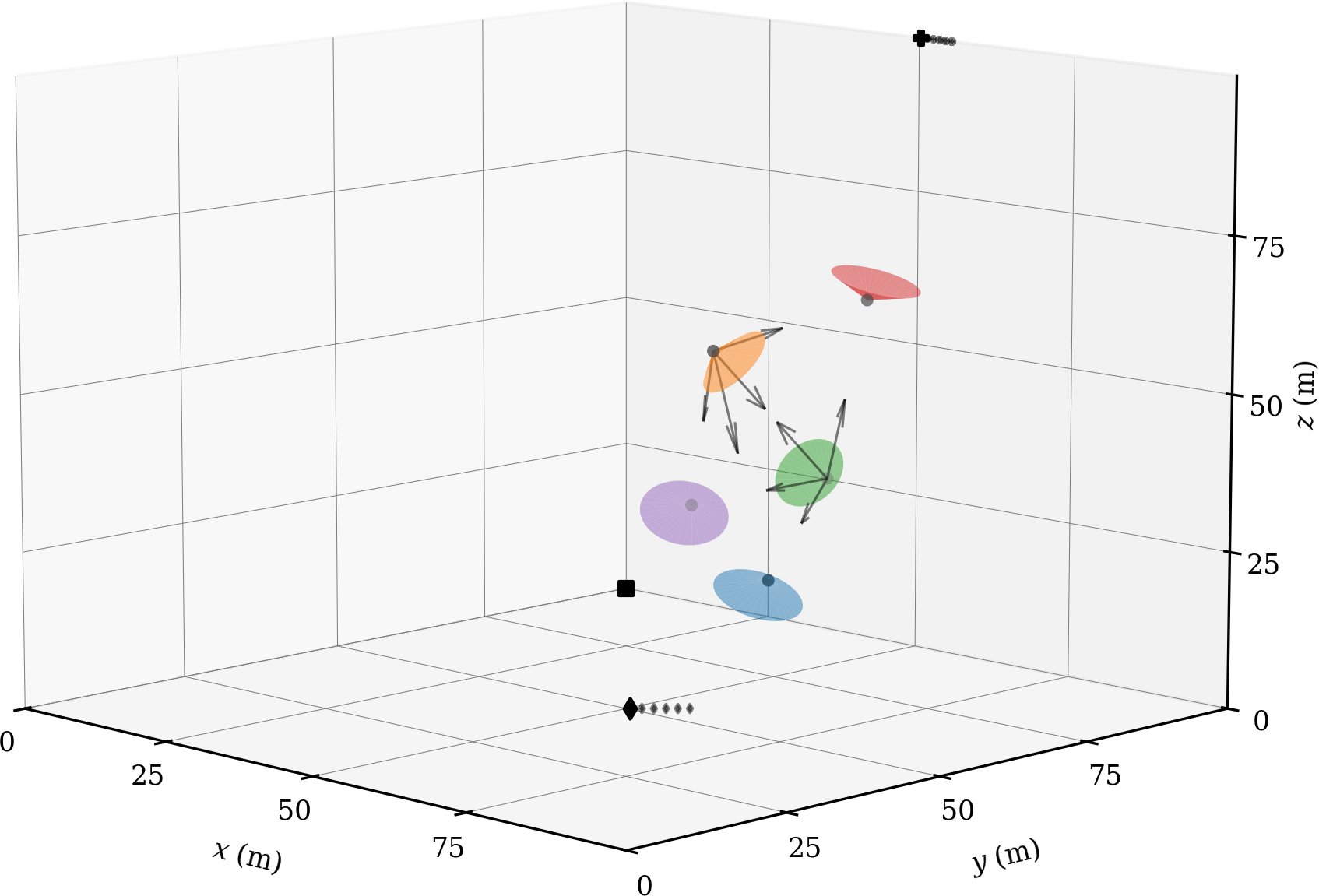}
        \caption{$t = \SI{150}{\second}$.}
    \end{subfigure}
    \caption{State of the system at different stages. Cameras are depicted as cones with apex located at $p_i$ and aiming at $R_i e_1$. The field of view $\phi$ is drawn to scale, while the range $\rho_r$ is not.
    Arrows represent acquired bearing measurements.}
    \label{fig:snapshots}
\end{figure}

\begin{figure}[!tb]
    \centering
    \begin{subfigure}{\columnwidth}
        \includegraphics[width=\columnwidth]{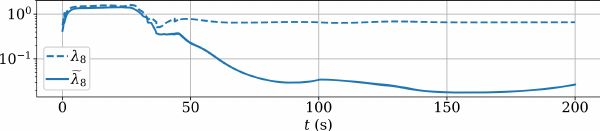}
        \caption{Rigidity eigenvalues: weighted $\widetilde{\lambda}_8$ and unweighted $\lambda_8$.}
        \label{fig:eigenvalues}
    \end{subfigure}
    \begin{subfigure}{\columnwidth}
        \includegraphics[width=\columnwidth]{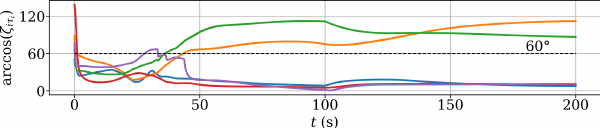}
        \caption{Aiming: an angle under $\SI{60}{\degree}$ indicates that the target is within the field of view of robot $i=1, \ldots, 5$.}
        \label{fig:targets}
    \end{subfigure}
    \begin{subfigure}{\columnwidth}
        \includegraphics[width=\columnwidth]{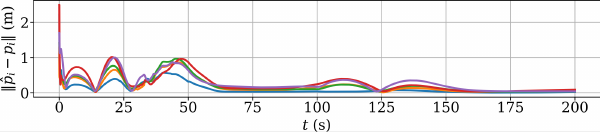}
        \caption{Localization error for robot $i=1, \ldots, 5$.}
        \label{fig:localization}
    \end{subfigure}
    \caption{Control performance metrics.}
    \label{fig:performance_metrics}
\end{figure}

\begin{figure}
    \centering
    \begin{subfigure}{0.49\columnwidth}
        \includegraphics[width=\columnwidth, trim=20mm 18mm 30mm 13mm, clip]{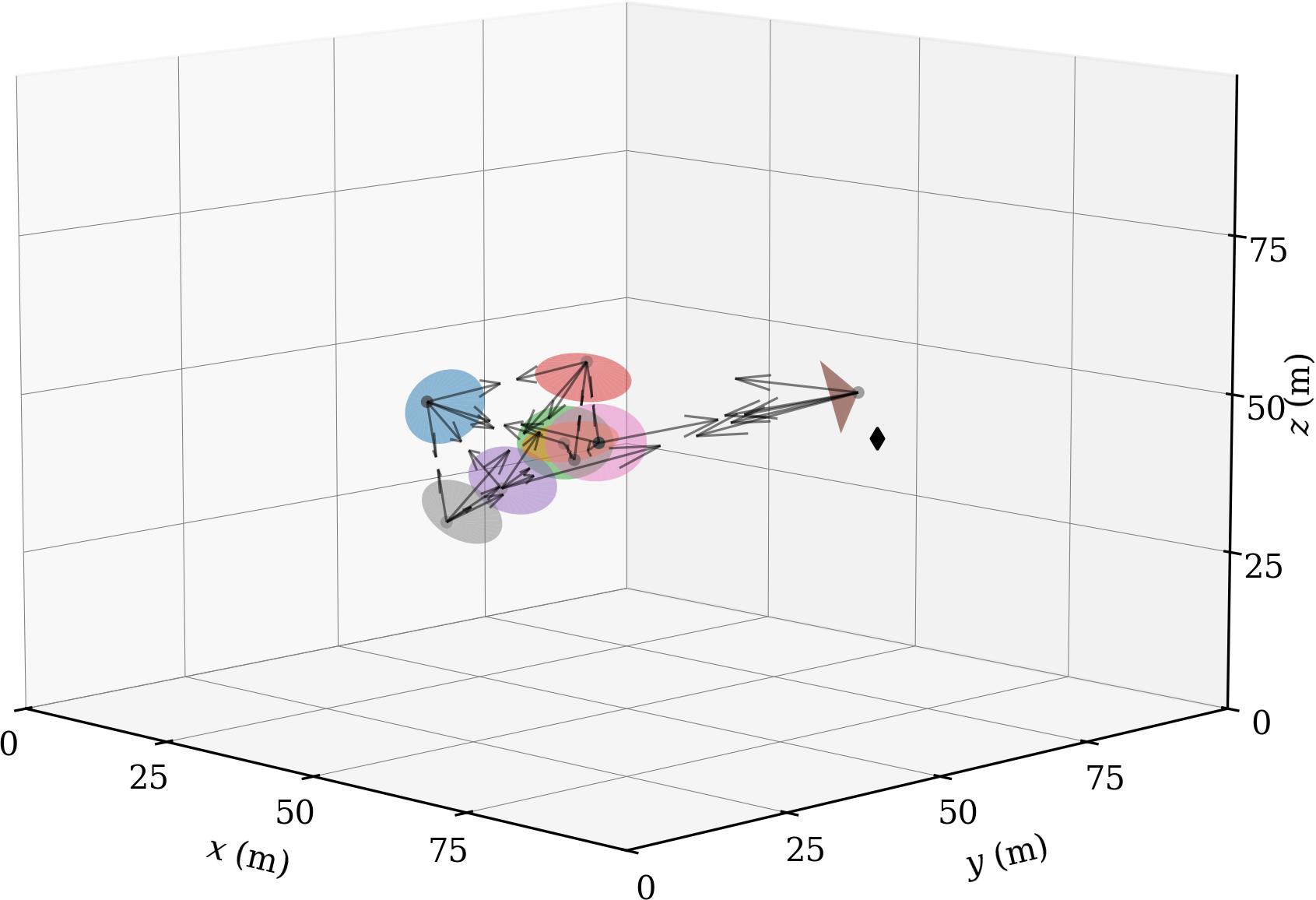}
        \caption{$t = \SI{0}{\second}$.}
    \end{subfigure}
    \begin{subfigure}{0.49\columnwidth}
        \includegraphics[width=\columnwidth, trim=20mm 18mm 30mm 13mm, clip]{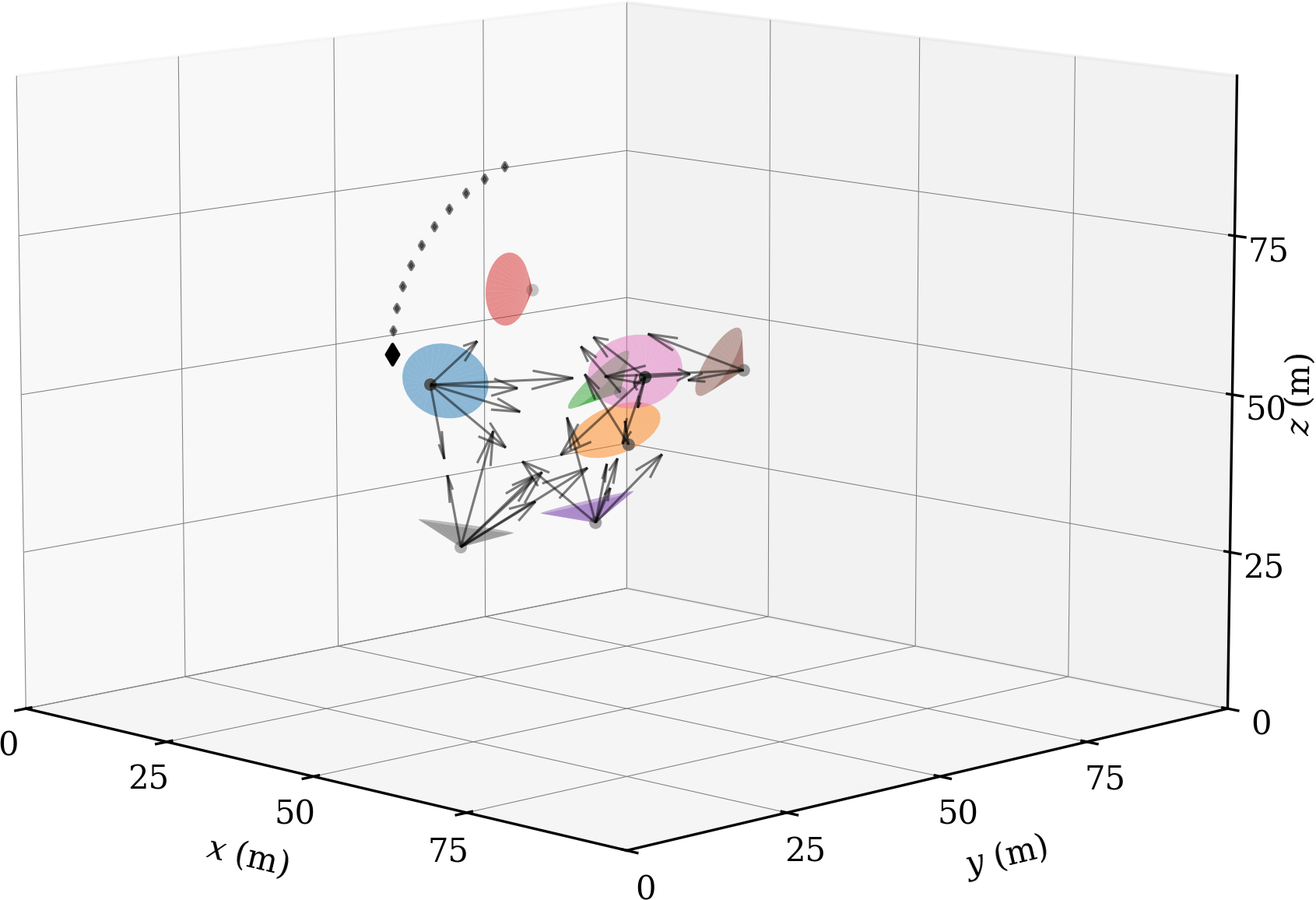}
        \caption{$t = \SI{25}{\second}$.}
    \end{subfigure}
    \begin{subfigure}{0.49\columnwidth}
        \includegraphics[width=\columnwidth, trim=20mm 18mm 30mm 13mm, clip]{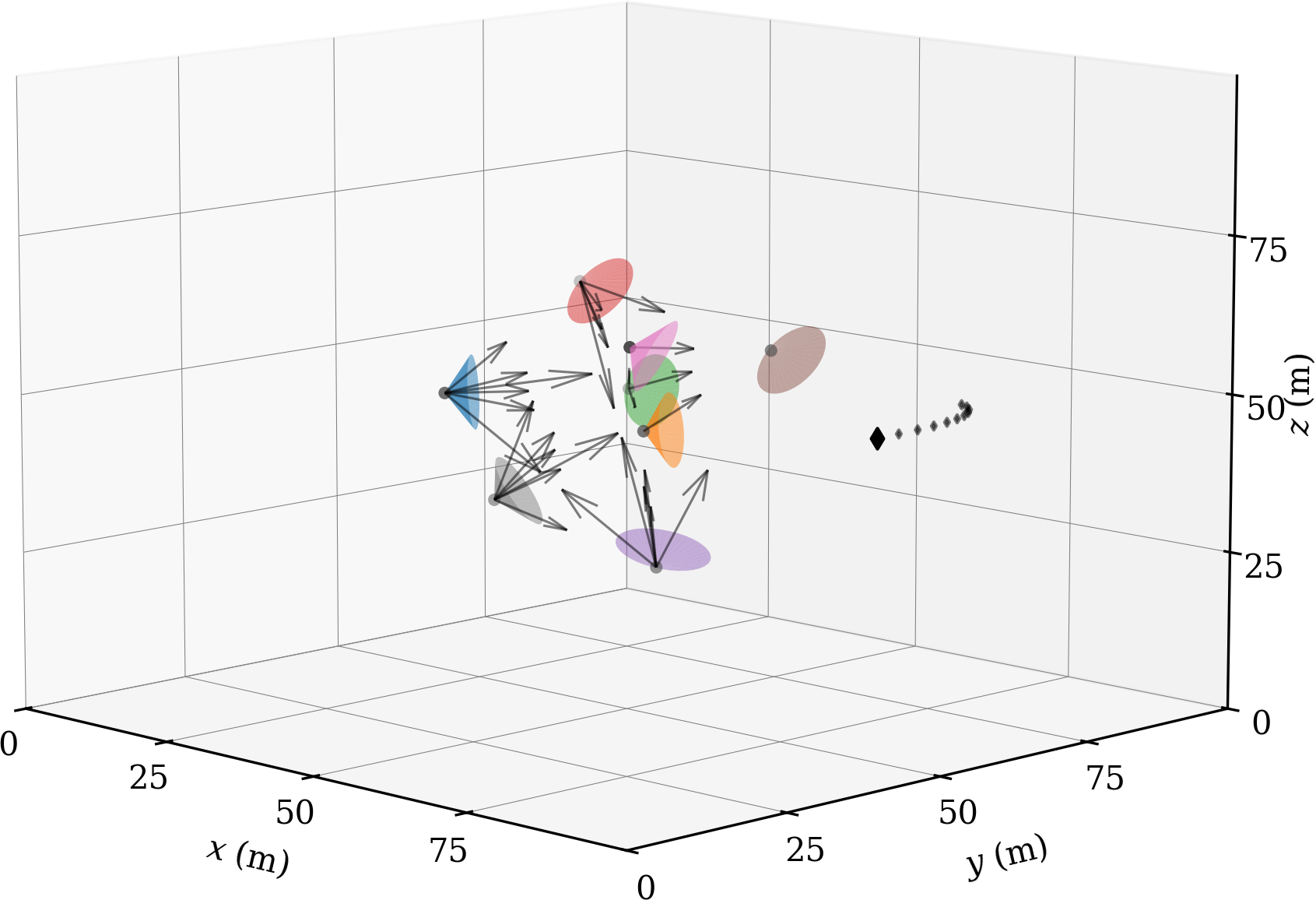}
        \caption{$t = \SI{50}{\second}$.}
    \end{subfigure}
    \begin{subfigure}{0.49\columnwidth}
        \includegraphics[width=\columnwidth, trim=20mm 18mm 30mm 13mm, clip]{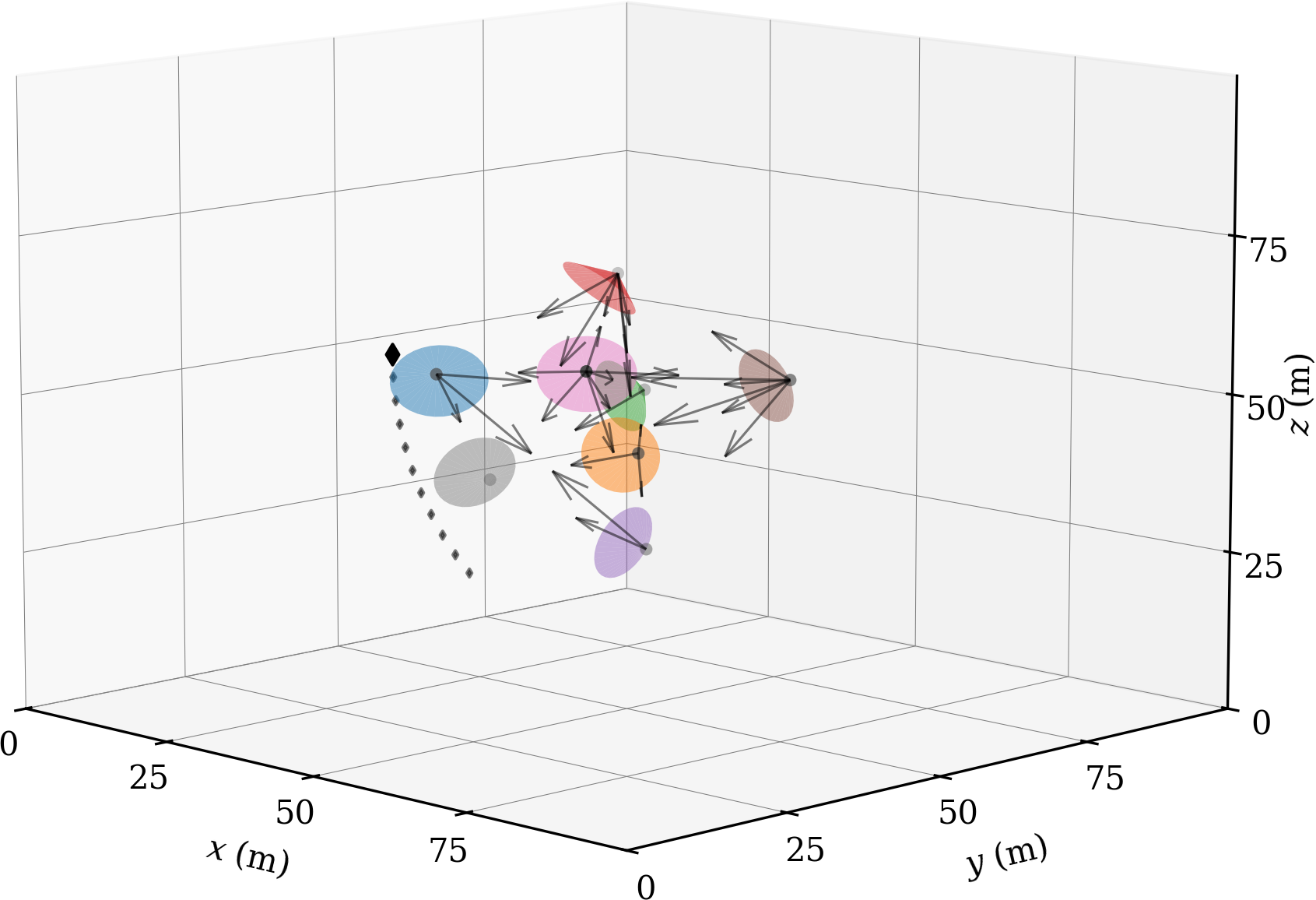}
        \caption{$t = \SI{75}{\second}$.}
    \end{subfigure}
    \caption{State of the system at different stages (second case).}
    \label{fig:snapshots_b}
\end{figure}

\begin{figure}[!tb]
    \centering
    \begin{subfigure}{\columnwidth}
        \includegraphics[width=\columnwidth]{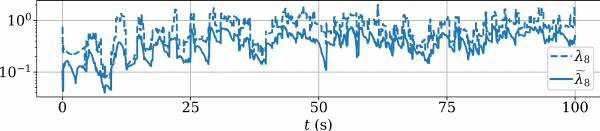}
        \caption{Rigidity eigenvalues: weighted $\widetilde{\lambda}_8$ and unweighted $\lambda_8$.}
    \end{subfigure}
    \begin{subfigure}{\columnwidth}
        \includegraphics[width=\columnwidth]{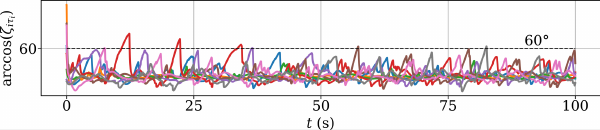}
        \caption{Aiming: an angle under $\SI{60}{\degree}$ indicates that the target is within the field of view of robot $i=1, \ldots, 5$.}
    \end{subfigure}
    \caption{Control performance metrics (second case).}
    \label{fig:performance_metrics_2}
\end{figure}

\section{Conclusions and Future Work}
\label{sec:conclusions}

This work studied angle-based localization and decentralized rigidity maintenance control for multi-robot networks with sensing constraints.
First, we provided a formal equivalence between bearing rigidity and angle rigidity that considers \textit{directed} sensing graphs and \textit{body-frame} bearing measurements in both $2$ and $3$-dimensional space.
Subsequently, we proposed an angle-based localization protocol for arbitrary graphs that only requires infinitesimal angle rigidity, and demonstrated its local convergence under time-varying topologies.

Additionally, the angle rigidity eigenvalue was proposed as a measure of a framework's degree of rigidity.
Closed-form expressions for the rigidity eigenvalue and its gradient were derived.
Subsequently, a decentralized rigidity maintenance controller was proposed and validated through simulations for the application of cooperative multi-target tracking.
This demonstrated the controller's ability to handle multiple objectives and sensing constraints while avoiding collisions and maintaining a sufficient level of angle rigidity.
This analysis revealed that, although relative rotations are not necessary for angle-based localization, they are required for rigidity maintenance, thus implying additional sensing capabilities.
Nevertheless, angle-based approach remain advantageous since they decouple the position estimation from the robots' orientations (a) liberating some robots from mutual sensing; (b) eliminating the introduction of orientation errors into the estimator.

Future work will continue along the following directions.
First, we plan to investigate the problem of estimating $R_{ij}$ without introducing additional sensing assumptions.
Second, it is interesting to explore methods that let robots dynamically select appropriate angle subsets from the observed bearings.
This would reduce angle redundancy, decreasing computation and communication load.
Finally, relying on the availability of the rigidity eigenvalue and eigenvector faces scalability issues; thus, it is worth exploring alternatives such as the use of \textit{subframeworks} proposed in \cite{PresenzaACC2022,Presenza2025CSL}.

\bibliographystyle{plain}        
\bibliography{refs}

\appendix
\section{Appendix}    

\subsection{A useful lemma}
\label{apx:lemmas}

\begin{lemma}
    Let $x, y_1, y_2, z_1, z_2 \in \mathbb{R}^3$ such that $z_1 \neq \pm z_2$, $\| z_1 \| = \| z_2 \| = 1$, $z_1 \perp y_1$, $z_2 \perp y_2$. The system of equations $S(z_1) x = y_1 \quad \text{and} \quad S(z_2) x = y_2$ with unknown $x$ has a solution if and only if $\left(S(z_1) y_1\!-\!S(z_2) y_2\right)^\top S(z_1) z_2 = 0$.
    \label{lem:linear_system_2}
\end{lemma}
\begin{proof}
    First, observe that $x^*_i \coloneqq \operatorname{span}\{z_i\} - S(z_i) y_i$ denotes the solution set of each respective sub-system.
    A common solution $x^\ast$ exists if and only if $S(z_1)y_1 - S(z_2)y_2 \in \operatorname{span}\{z_1, z_2\}$.
    Due to the linear independence between $z_1$ and $z_2$, this is equivalent to $S(z_1) y_1 -  S(z_2) y_2 \ \perp \ S(z_1) z_2$, and the claim follows.
\end{proof}

\subsection{Gradient computation on $\mathrm{SO}(d)$}
\label{apx:rigidity_control}

Here we show how to compute the gradient of a quadratic function $f : \mathrm{SO}(d) \to \mathbb{R}$, $f(R) = x^\top R y$, with respect to the Riemannian metric $\langle A, B \rangle \coloneqq \frac{1}{2} \operatorname{Tr}(A^\top B)$.
This is a required step to obtain \eqref{eq:grad_w_2}. To do that, consider
\begin{equation*}
    g :\mathbb{R}^{d \times d} \to \mathbb{R}, \ g(A) = x^\top A y \implies \nabla_A g(A) = 2 x y^\top.
\end{equation*}
Then,
\begin{equation*}
    \nabla_R f(R) = R \operatorname{skew}\left(R^\top \nabla_A g(R)\right)
\end{equation*}
where $\mathrm{skew}(A) \coloneqq \frac{1}{2}\left(A - A^\top\right)$, see \cite{Absil2008Princeton}.
Therefore,
\begin{equation*}
    R^\top \nabla_R \left(x^\top R y\right)\!=\!\operatorname{skew}\!\left(R^\top 2 x y^\top \right)\!=\!\left(R^\top x y^\top\!-\!y x^\top R \right).
\end{equation*}

\end{document}